%% file: main.tex
\documentclass[11pt,letterpaper]{article}

\usepackage[T1]{fontenc}
\usepackage{paper_style}
\usepackage{macros}

\newif\ifanonymous
\anonymousfalse

\title{Tight Lower Bounds for State Tomography\\
with Limited Entanglement}

\ifanonymous
  \author{}
\else
  \author{Ufuk Keskin\\
MIT\\
\texttt{
ufukkesk@mit.edu}
\and
Jason Luo\\
MIT\\
\texttt{luojason@mit.edu}
\and
Mahbod Majid\thanks{Supported by the 2026 Apple Scholars in AI/ML PhD fellowship.}\\
MIT\\
\texttt{mahbod@mit.edu}
\and
Matthew Radzihovsky\\
MIT\\
\texttt{mattradz@mit.edu}}
\fi
\date{}

\hypersetup{
  pdftitle={Tight Lower Bounds for State Tomography with Limited Entanglement},
}

\begin{document}

\hypersetup{pageanchor=false}
\maketitle

\begin{abstract}
We study state tomography when each measurement acts on at most $k$
fresh copies and no quantum memory is retained between blocks.  We prove a
lower bound matching the upper bound in \cite{PSW26}.  Thus the copy complexity
of estimating an arbitrary $d$-dimensional state to trace distance $\eps$ is,
up to absolute constant factors,
$\max\{d^3/(\sqrt{k}\eps^2),d^2/\eps^2\}$ for every $k$ and all sufficiently
small $\eps$.  This removes
the earlier restriction that $k$ be small as a function of the accuracy.  The
lower bound applies to arbitrary measurements within each block and adaptive
choices between blocks.  The lower bound already applies in a small neighborhood
of any state whose smallest eigenvalue is of order $1/d$, even when the center is
known.

The main ingredient is a uniform Fisher information bound for one measurement
block that depends only on the smallest eigenvalue of the state.  The proof
avoids the perturbative expansion responsible for the restriction in
\cite{CLL24}.  Fano's inequality for metric balls and a log-Sobolev comparison
between mutual and Fisher information then reduce the adaptive protocol to
this block bound \cite{XuRaginsky17,ALPC19}.
\end{abstract}

\thispagestyle{empty}
\clearpage

\microtypesetup{protrusion=false}
\thispagestyle{empty}
\tableofcontents
\thispagestyle{empty}
\microtypesetup{protrusion=true}
\clearpage

\setcounter{page}{1}
\hypersetup{pageanchor=true}

\input{sections/introduction}
\input{sections/technical_overview}
\input{sections/related_work}
\input{sections/preliminaries}
\input{sections/main_lower_bound}
\input{sections/anti_concentration}
\input{sections/global_information}
\input{sections/block_fisher_information}

\noindent\textbf{Statement on AI use.}
The main ideas and proof strategies were developed with assistance from OpenAI
Codex using the GPT-5.6 Sol model.  The authors studied, refined, simplified,
and verified the resulting material and take responsibility for every claim,
proof, and citation.

\clearpage
\begingroup
\small
\bibliographystyle{alpha}
\bibliography{references}
\endgroup

\end{document}

%% file: sections/introduction.tex
\section{Introduction}
\label{sec:introduction}

State tomography is the task of learning a
density operator describing an unknown quantum state from identically prepared copies of the quantum state. The copy complexity of state tomography
depends not only on the number of state parameters, but also on how many copies may
be measured jointly.  A quantum measurement acting jointly on several copies is
called a collective measurement.  In the unrestricted version of this model,
one joint POVM may act on all available copies, and the copy complexity is of
order $d^2/\eps^2$ for learning an arbitrary $d$-dimensional mixed state to
trace distance $\eps$~\cite{HHJWWY16,OW16}.  At the other extreme, a procedure
restricted to measuring one copy at a time requires order $d^3/\eps^2$
copies, even if its measurements are chosen
adaptively~\cite{GKKT20,CHHLLS23}.

These endpoints motivate the intermediate model in which each measurement
acts on at most $k$ copies and the protocol retains only classical information
between blocks of measurements.  The problem is to determine how the copy complexity interpolates between
single-copy and unrestricted joint measurements.

In \cite{CLL24}, Chen, Li, and Liu established the first tradeoff between the
number of copies measured jointly and the copy complexity of tomography.  They
proved a lower bound of order
$d^3/(\sqrt{k}\eps^2)$ when $k$ is at most a sufficiently small absolute
power of $1/\eps$, together with a matching upper bound up to logarithmic
factors over the corresponding range.  More recently, in
\cite{PSW26}, Pelecanos, Spilecki, and Wright gave, for every $k$, the following
upper bound without logarithmic factors:
\[
  n\lesssim
    \max\Set{
      \frac{d^3}{\sqrt{k}\eps^2},
      \frac{d^2}{\eps^2}
    }.\footnotemark
\]
\footnotetext{For nonnegative quantities $X$ and $Y$, we write
$X\lesssim Y$ if $X\leq CY$ for an absolute constant $C>0$.  We write
$X\gtrsim Y$ if $Y\lesssim X$, and $X\asymp Y$ if both relations hold.}
They conjectured that this rate is optimal for every $k$ and left the matching
lower bound as an open problem.  This leads us to the
following central question:

\begin{quote}
\begin{center}
  \emph{What is the optimal copy complexity when each measurement acts on at
  most $k$ copies?}
\end{center}
\end{quote}

\subsection{Our Results}

We determine the optimal copy complexity and fully resolve the above question.
For every $k$, we prove a lower bound matching the algorithm in \cite{PSW26},
without imposing any relation between $k$ and $\eps$.  Each measurement may be
arbitrary on its block, and later measurements may depend on every earlier
classical outcome.

We begin with the statistical and measurement models.  A state on
$\C^d$ is a positive semidefinite matrix $\rho\in\C^{d\times d}$ with
$\Tr(\rho)=1$, and $I_d$ denotes the identity on $\C^d$.  The trace distance
between two states is
\[
  \Dtr(\rho,\sigma)=\frac12\norm{\rho-\sigma}_1.
\]
For a Hermitian matrix $H$, $\norm H_{\op}$ is its largest absolute
eigenvalue.
Any quantum measurement performed on $t$ copies with a finite outcome set $\mathcal Z$ is described by a positive operator-valued measure (POVM)
on $t$ copies. This is a collection of positive semidefinite operators
$\{M_z:z\in\mathcal Z\}$ on $(\C^d)^{\otimes t}$ such that
$\sum_zM_z=I_d^{\otimes t}$.  When such a quantum measurement is performed on $\rho^{\otimes t}$, the
outcome $z$ is obtained with probability $\Tr(M_z\rho^{\otimes t})$.
We use finite-outcome notation for brevity. For a general standard Borel outcome space, the formal argument
applies the finite-outcome result to every measurable function of the outcome
with finite range; see \cref{sec:prelim-fisher}.

\begin{definition}[Protocols with measurements on at most $k$ copies]
\label{def:block-protocol-intro}
For $d,k,n\in\N$, a protocol in this model receives $n$ independent copies
of an unknown state $\rho\in\C^{d\times d}$.  In round $j$, as a function of
its internal randomness and the previous classical outcomes, it chooses an
integer $1\leq t_j\leq k$ and applies an arbitrary joint quantum measurement to $t_j$ fresh copies
of $\rho$.  The protocol may stop at a
random time, but must satisfy
$\sum_{j=1}^{T} t_j\leq n$ almost surely, where $T$ is the number of rounds. 
The protocol may adaptively choose the POVM describing the quantum measurement performed on the fresh copies depending on the previous measurement outcomes.
It may retain unlimited classical
information between rounds, but it retains no quantum memory from
one round to the next.
\end{definition}

\begin{problem}[Tomography with limited entanglement]
\label{prob:tomography-intro}
Given the transcript of a protocol in \cref{def:block-protocol-intro}, output
a state $\widehat\rho$ such that
\[
  \Dtr(\rho,\widehat\rho)\leq\eps
\]
with probability at least $2/3$, for every input state $\rho$.
\end{problem}

Our main result is the following.

\begin{restatable}[Tight lower bound for tomography with limited entanglement]
  {theorem}{mainlowerbound}
\label{thm:main-lower-bound}
There is an absolute constant $\eps_0>0$ such that the following holds.
For every $d\geq2$, every integer $k\geq1$, and every
$0<\eps\leq\eps_0$, any protocol in \cref{def:block-protocol-intro} that solves
\cref{prob:tomography-intro} using at most $n$ copies satisfies
\begin{equation}
\label{eq:main-lower-bound}
  n
  \gtrsim
  \frac{d^3}{\eps^2\sqrt{\min\{k,d^2\}}}
  =
  \max\Set{
    \frac{d^3}{\sqrt{k}\eps^2},
    \frac{d^2}{\eps^2}
  }.
\end{equation}
The bound holds for every protocol in \cref{def:block-protocol-intro},
including adaptive measurements and nonuniform or random block sizes.
\end{restatable}

The same rate is already necessary in a small neighborhood of any state whose
smallest eigenvalue is of order $1/d$.  The center may be known to the
protocol.

\begin{corollary}[Local lower bound]
\label{cor:local-lower-bound}
There are absolute constants $L\geq4$ and $\eps_1>0$ such that the following
holds.  Let $d\geq2$, let $k\geq1$, and let $\rho_\star$ be a state satisfying
$\rho_\star\succeq I_d/(2d)$.  For every $0<\eps\leq\eps_1$, any protocol that,
even knowing $\rho_\star$, estimates every state $\rho$ satisfying
\[
  \norm{\rho-\rho_\star}_{\op}\leq\frac{L\eps}{d}
\]
to trace distance $\eps$ with probability at least $2/3$ must use
\[
  n
  \gtrsim
  \frac{d^3}{\eps^2\sqrt{\min\{k,d^2\}}}
\]
copies.  Measurements and adaptivity are as in
\cref{def:block-protocol-intro}.
\end{corollary}

The almost sure copy budget in the theorem can also be replaced by a uniform
bound in expectation.

\begin{corollary}[Expected copy budget]
\label{cor:expected-copy-budget}
The conclusions of \cref{thm:main-lower-bound,cor:local-lower-bound} remain
valid if the protocol has a random total number $N$ of copies and satisfies
\[
  \sup_\rho\E_\rho N\leq n,
\]
where $\E_\rho$ denotes expectation over the protocol's internal randomness
and measurement outcomes when the input state is $\rho$.  The supremum ranges
over the states in the corresponding estimation problem.
\end{corollary}

The proof of the main lower bound gives the stronger requirement in
\cref{eq:weighted-copy-requirement}, expressed in terms of the sizes of the
blocks used by the protocol; it therefore applies directly to nonuniform and
random block sizes.

The upper-bound estimator of \cite{PSW26} is unbiased, but its
Hermitian, trace-one output may have negative eigenvalues.  No estimator can
always output a state and remain unbiased for every input: for a pure input,
unbiasedness would force the output to equal that state almost surely.
Averaging the unbiased block estimates centers the final average at the true
state.  To satisfy \cref{prob:tomography-intro}, replace only that final average
by a closest state in trace norm, as observed in \cite[Remark~2.10]{PSW26}.
Since the true state is feasible, the triangle inequality increases the error
by at most a factor of two.  Running their estimator with error target
$\eps/2$ therefore changes the copy bound only by an absolute constant.
Combining this upper bound with \cref{thm:main-lower-bound} determines the
optimal copy complexity up to absolute constants.

\begin{corollary}[Optimal copy complexity for measurements on at most $k$ copies]
\label{cor:complete-tradeoff}
For every $d\geq2$, every integer $k\geq1$, and every
$0<\eps\leq\eps_0$, the minimax number of copies required by protocols in
\cref{def:block-protocol-intro} to learn every $d$-dimensional state to trace
distance $\eps$ with success probability at least $2/3$, denoted by
$n^\star(d,k,\eps)$, satisfies
\[
  n^\star(d,k,\eps)
  \asymp
    \max\Set{
      \frac{d^3}{\sqrt{k}\eps^2},
      \frac{d^2}{\eps^2}
    }.
\]
\end{corollary}

At $k=1$, the first term matches both the single-copy upper bound and the
adaptive lower bound~\cite{GKKT20,CHHLLS23}.  For
$1<k<d^2$, it captures the $\sqrt{k}$ gain from joint measurements.  Once
$k\geq d^2$, the second term matches the lower bound for unrestricted
measurements~\cite{HHJWWY16}.  Thus the new content of
\cref{thm:main-lower-bound} is the intermediate regime, where it removes the
restriction relating $k$ to the accuracy in the lower bound of \cite{CLL24}.

\paragraph{Organization.}
\Cref{sec:technical-overview} gives an overview of the proof.  \Cref{sec:related-work} reviews related
results.  We then restate the known
facts we need from statistics, analysis, and representation theory and connect
them to our setting in \cref{sec:preliminaries}.
\Cref{sec:main-lower-bound-proof} derives the main
and local lower bounds
(\cref{thm:main-lower-bound,cor:local-lower-bound}) from the two information
bounds.
\Cref{sec:anti-concentration} constructs the prior and proves the mutual
information lower bound required for tomography
(\cref{cor:information-lower-bound}).
\Cref{sec:global-information} proves the mutual information upper bound for
adaptive protocols (\cref{prop:information-upper-bound}), and
\cref{sec:block-fi} proves the Fisher information bound for one measurement
block (\cref{thm:block-fi}).

%% file: sections/technical_overview.tex
\section{Technical Overview}
\label{sec:technical-overview}
Theorem~\ref{thm:main-lower-bound} is a lower bound on the number of copies.
We prove it by comparing two bounds on the same mutual information.
Draw a perturbation $\Delta$ of a known reference state from a truncated
Gaussian prior, and let $Y$ be the full classical transcript of the protocol.
Successful tomography gives the mutual information lower bound
\[
  I(\Delta;Y)\gtrsim d^2
\]
through a metric form of Fano's inequality.  In the other direction, a
log-Sobolev comparison reduces the mutual information $I(\Delta;Y)$ to the
Fisher information of the transcript, and the Fisher information chain rule
reduces the latter to the contributions of the individual measurement
blocks.  An upper bound on the Fisher information of one block then bounds
$I(\Delta;Y)$ in terms of the number of copies $n$ and the maximum block size
$k$.  Comparing these two bounds on $I(\Delta;Y)$ gives the copy complexity
lower bound.

The information-theoretic reduction uses known ingredients: the truncated
Gaussian prior from \cite{CHHLLS23,CLL24}, the metric form of Fano's inequality
from \cite[Lemma~1]{XuRaginsky17}, the comparison between mutual and Fisher
information from \cite[Theorem~1]{ALPC19}, and the standard Fisher information
chain rule for adaptive experiments.  The new technical contribution is the
$t^{3/2}/\lambda_{\min}(\rho)^2$ branch of the Fisher information bound for an
arbitrary $t$-copy measurement.  The change of variables used to express each
outcome score has a two-copy antecedent in
\cite[Supplemental Material, Section~VI, Eqs.~(S63)--(S66)]{ZhuHayashi18},
and the two Schur--Weyl estimates used to control the resulting unfolded
matrix come from
\cite[full version, Lemma~6.1 and the proof of Claim~3.27]{CLL24}.
The step specific to our argument is to use POVM completeness to identify the mean of the unfolded matrix as $t\rho$ and to center at this mean before
applying those estimates.  This cancels the term of order
$t^2\Tr(\rho^2)$ and leaves the required contribution of order $t^{3/2}$.
The resulting calculation avoids the perturbative expansion of
$\rho^{\otimes t}$ used in \cite{CLL24} and therefore imposes no relation
between the block size and the accuracy.  The complementary bound for larger
blocks is the standard consequence of the symmetric logarithmic derivative
quantum Fisher information bound and tensor additivity.  Translating the same
prior to a known well-conditioned state gives the local lower bound in
\cref{cor:local-lower-bound}.
\subsection{The Main Theorem from Two Information Bounds}
\label{sec:overview-information-comparison}

We instantiate this strategy with a truncated Gaussian prior, state the two
information bounds, and compare them.

Let $\operatorname{Herm}_0(\C^d)$ denote the set of traceless Hermitian
$d\times d$ matrices, viewed as a Euclidean space of dimension $d^2-1$ with
inner product $\ip HG=\Tr(HG)$; see \cref{sec:prelim-matrix} for the formal
definition and coordinate conventions.  We write
$\norm X_F=(\Tr(X^\dagger X))^{1/2}$ for the Frobenius norm and
$\norm X_{\op}$ for the operator norm.  Let $Z$ be a standard Gaussian on
$\operatorname{Herm}_0(\C^d)$.  Explicitly, for any orthonormal basis
$F_1,\ldots,F_{d^2-1}$ of this space,
\[
  Z=\sum_{b=1}^{d^2-1}g_bF_b,
  \qquad
  g_1,\ldots,g_{d^2-1}\ \text{are independent }\mathcal N(0,1)
  \text{ random variables}.
\]
Condition on
\[
  \norm Z_{\op}\leq\frac{\Cprior}{\sqrt2}\sqrt d,
\]
and set
\[
  \tau=\Cprior\eps,
  \qquad
  \Delta=\frac{\sqrt2\,\tau}{d^{3/2}}Z.
\]
Here $\Cprior$ is the absolute constant from the uniform prior mass bound
(\cref{prop:prior-properties}).  We denote the conditional law of $\Delta$ by
$\mu$.  Equivalently, with
\[
  \kappa=\frac{d^3}{2\tau^2},
  \qquad
  K=\Set{
    \Delta\in\operatorname{Herm}_0(\C^d):
    \norm\Delta_{\op}\leq\frac{\Cprior\tau}{d}},
\]
its density is
\begin{equation}
\label{eq:overview-prior}
  \dd\mu(\Delta)
  \propto
  \exp\Paren{-\frac{\kappa}{2}\norm\Delta_F^2}
  \one_K(\Delta)\dd\Delta.
\end{equation}
We draw $\Delta\sim\mu$, run the protocol on
$\rho_\Delta=I_d/d+\Delta$, and let $Y$ be its full classical transcript.
Fano's inequality for metric balls shows that successful tomography requires
$I(\Delta;Y)\gtrsim d^2$ (\cref{cor:information-lower-bound}).  Conversely,
the Fisher information argument gives
$I(\Delta;Y)\lesssim
(\eps^2n/d)\sqrt{\min\{k,d^2\}}$
(\cref{prop:information-upper-bound}).  Combining the two estimates gives
\begin{equation}
\label{eq:overview-information-sandwich}
  \underbrace{d^2}_{\text{required for estimation}}
  \lesssim
  I(\Delta;Y)
  \lesssim
  \underbrace{\frac{\eps^2n}{d}
    \sqrt{\min\{k,d^2\}}}_{\text{available from the transcript}}.
\end{equation}
Rearranging yields
\[
  n
  \gtrsim
  \frac{d^3}{\eps^2\sqrt{\min\{k,d^2\}}}.
\]

Thus the main proof is reduced to the two sides of this sandwich.  They
prove the desired lower bound for all sufficiently large $d$; the formal
proof in \cref{sec:main-lower-bound-proof} uses the standard lower bound for
unrestricted joint measurements (\cref{cor:fixed-dimensional-patch}) for the
finitely many smaller dimensions.

For the local lower bound in \cref{cor:local-lower-bound}, use
$\rho_\star+\Delta$ in place of $I_d/d+\Delta$, where the protocol knows
$\rho_\star\succeq I_d/(2d)$.  The support of $\mu$ lies in the prescribed
neighborhood of $\rho_\star$ and, for sufficiently small $\eps$, satisfies
$\rho_\star+\Delta\succeq I_d/(4d)$.  The law of $\Delta$, and hence its
log-Sobolev and small-ball bounds, is unchanged.  The same comparison proves
the result in large dimensions; \cref{sec:main-lower-bound-proof} gives the
argument using two states for the remaining dimensions.

\subsection{The Mutual Information Upper Bound from One Block}
\label{sec:overview-information-upper}

So far, we have shown that the main theorem follows from lower and upper
bounds on $I(\Delta;Y)$.  
We first prove the upper bound on
$I(\Delta;Y)$: a log-Sobolev comparison converts mutual information to Fisher
information, after which the Fisher information chain rule reduces the
adaptive transcript to its individual measurement blocks.

Let $F_1,\ldots,F_{d^2-1}$ be an orthonormal basis of the traceless Hermitian
matrices and use the local coordinates
\[
  \rho_\theta=\rho+\sum_{b=1}^{d^2-1}\theta_bF_b.
\]
If an observation $W$ has likelihood $p_\theta(w)$ under $\rho_\theta$, its
score at an outcome $w$ is
\[
  s_W(w)=\left.\nabla_\theta\log p_\theta(w)\right|_{\theta=0}.
\]
The Fisher information $J_W(\rho)$ is the second moment of the score, so
$\Tr J_W(\rho)=\E_\rho\norm{s_W(W)}_2^2$.  Normalization gives
$\E_\rho s_W(W)=0$.  Here and below, $\E_\rho$ averages over all
measurement outcomes and protocol randomness when the input state is $\rho$.

In \cite[Theorem~1]{ALPC19}, Aras, Lee, Pananjady, and Courtade prove a general
inequality that bounds mutual information through a log-Sobolev inequality and
average Fisher information.  We state the specialization needed here in
\cref{fact:lsi-to-mutual-information} and derive its application to the full
transcript in \cref{eq:mi-from-fi}.  For the prior in
\cref{eq:overview-prior}, this gives
\begin{equation}
\label{eq:overview-lsi-conversion}
  I(\Delta;Y)
  \lesssim
  \frac{\eps^2}{d^3}
  \E_{\Delta\sim\mu}\Tr J_Y(\rho_\Delta).
\end{equation}
It therefore remains to control the Fisher information of the transcript at
each $\rho_\Delta$ with $\Delta\in K$.

A direct argument using mutual information would analyze one block at a time.
After each outcome, however, the law of $\Delta$ changes, so the same bound
need not apply to the next block.  The problem is the changing prior, not the
blocks themselves.  We instead fix the state and use Fisher information.  The
transcript likelihood factors across rounds, so its score is the sum of the
conditional scores.  These scores have conditional mean zero because each
conditional outcome distribution has total probability one.  Hence the cross
terms vanish and Fisher information adds.  The Fisher information chain rule
for an adaptive transcript (\cref{fact:adaptive-fi-composition}) therefore
states that, if the preceding outcomes and protocol randomness select $M_j$
on $t_j$ copies in round $j$, then
\[
  \Tr J_Y(\rho)
  =
  \E_\rho\Brac{\sum_{j=1}^{T}\Tr J_{M_j}(\rho)}.
\]

The remaining input is the Fisher information bound for a block measurement
(\cref{thm:block-fi}).  We explain its proof in
\cref{sec:overview-block-fi}.  Its formal statement and full proof appear in
\cref{sec:block-fi}.  For every full-rank state $\rho$, the theorem states
that
\[
  \Tr J_M(\rho)
  \leq
  \min\Set{
    \frac{2t^{3/2}}{\lambda_{\min}(\rho)^2},
    \frac{t(d^2-1)}{\lambda_{\min}(\rho)}
  }.
\]
In particular, if $\rho\succeq I_d/(4d)$, every measurement $M$ on $t$
copies satisfies
\begin{equation}
\label{eq:overview-block-bound}
  \Tr J_M(\rho)
  \lesssim
  \min\Set{d^2t^{3/2},d^3t}
  =
  d^2t\sqrt{\min\{t,d^2\}}.
\end{equation}
Since $t_j\leq k$ and $\sum_{j=1}^{T}t_j\leq n$ almost surely, substituting
\cref{eq:overview-block-bound} into the Fisher information chain rule above
gives the adaptive bound
(\cref{eq:adaptive-block-fi}), proved in
\cref{sec:global-information}:
\begin{align}
\label{eq:overview-adaptive-bound}
  \Tr J_Y(\rho)
  &\lesssim
  d^2\E_\rho\Brac{
    \sum_{j=1}^{T}t_j\sqrt{\min\{t_j,d^2\}}}
  \notag\\
  &\leq
  d^2n\sqrt{\min\{k,d^2\}}.
\end{align}
Substituting the adaptive Fisher information bound into the log-Sobolev
comparison gives
\[
  I(\Delta;Y)
  \lesssim
  \frac{\eps^2n}{d}\sqrt{\min\{k,d^2\}}.
\]
This is the mutual information upper bound for adaptive protocols
(\cref{prop:information-upper-bound}).  It also allows random block sizes,
adaptive measurements, and early stopping.

It remains to prove the block bound and verify that the prior from
\cite{CHHLLS23,CLL24} has the properties required by the two information
bounds.

\subsection{The Fisher Information in One Block}
\label{sec:overview-block-fi}

So far, we have reduced the information upper bound to the Fisher information
of a single measurement block.  We now prove the required block bound.  The
main issue is to improve the standard estimate when $t\leq d^2$.

The Braunstein--Caves measurement inequality and tensor additivity
(\cref{fact:quantum-fisher-information}) give the standard bound
\[
  \Tr J_M(\rho)
  \leq
  \frac{t(d^2-1)}{\lambda_{\min}(\rho)}.
\]
When $\rho\succeq I_d/(4d)$, this is of order $d^3t$, which is the desired
block bound when $t>d^2$.  Why is the standard estimate not enough for smaller
$t$?  Summing it over the blocks gives
$\Tr J_Y(\rho)\lesssim d^3n$.  After
\cref{eq:overview-lsi-conversion}, this recovers the known lower bound
$n\gtrsim d^2/\eps^2$ for unrestricted measurements~\cite{HHJWWY16}.  It gives
no dependence on $k$.  The new estimate is
\[
  \Tr J_M(\rho)
  \leq
  \frac{2t^{3/2}}{\lambda_{\min}(\rho)^2}.
\]
For $\rho\succeq I_d/(4d)$, this is of order $d^2t^{3/2}$, the improvement
needed when $t\leq d^2$.  The complete proof is in \cref{sec:block-fi}.

The argument below assigns a unit vector to each rank-one effect.  For a
finite POVM $M=\{M_y\}_y$, decompose each effect spectrally and report both
indices, where $w_{y,j}\geq0$ and each $v_{y,j}$ is a unit vector:
\[
  M_y=\sum_jw_{y,j}\proj{v_{y,j}},
  \qquad
  \widetilde M_{y,j}=w_{y,j}\proj{v_{y,j}},
  \qquad
  \sum_{y,j}\widetilde M_{y,j}=\sum_yM_y=I_d^{\otimes t}.
\]
Thus $\widetilde M$ is a POVM.  Forgetting $j$ recovers the outcome law of
$M$, so data processing gives
$J_M(\rho)\preceq J_{\widetilde M}(\rho)$
(\cref{fact:fi-data-processing}); it therefore suffices to consider rank-one
effects $w_z\proj{v_z}$.  General outcome spaces follow by applying the
finite-outcome argument to every measurable finite-valued function of the
outcome and then using \cref{eq:general-outcome-fi-supremum}.

For two copies, the change of variables below appears in
\cite[Supplemental Material, Section~VI, Eqs.~(S63)--(S66)]{ZhuHayashi18}.
The same calculation applies to every block size; its formal statement is
\cref{lem:block-score-ensemble}.

\paragraph{Scores of the measurement outcomes and the pure-state decomposition.}
So far, we have reduced the proof of
$\Tr J_M(\rho)\leq2t^{3/2}/\lambda_{\min}(\rho)^2$ to a finite POVM with
rank-one effects $w_z\proj{v_z}$.  We now rewrite the score of each outcome
using a vector $u_z$.  The pairs $(p_z,u_z)$ will form a pure-state
decomposition of $\rho^{\otimes t}$, allowing the Schur--Weyl estimates to
control the average squared score.

Let $v_z$ be a unit vector and set $R=\rho^{\otimes t}$.  Each outcome $z$
determines
\begin{equation}
\label{eq:overview-rank-one-variables}
  q_z=\bra{v_z}R\ket{v_z},
  \qquad
  p_z=w_zq_z,
  \qquad
  \ket{u_z}=\frac{R^{1/2}\ket{v_z}}{\sqrt{q_z}}.
\end{equation}
Then $p_z$ is the probability of outcome $z$, while $u_z$ is introduced only
for the proof.

For a unit vector $u\in(\C^d)^{\otimes t}$, let
\[
  G_1(u)
  =
  \sum_{\ell=1}^t
  \Tr_{\{1,\ldots,t\}\setminus\{\ell\}}(\proj u).
\]
Here $\Tr_S$ denotes the partial trace over the tensor factors indexed by
$S$.  Thus the $\ell$th summand is the reduced state on the $\ell$th copy,
viewed as an operator on $\C^d$.  Following \cite{CLL24}, we call $G_1(u)$
the \emph{unfolded matrix}.  It is the sum of these $t$ reduced states, and hence
$\Tr G_1(u)=t$.  For each outcome, write
\[
  B_z=G_1(u_z)-t\rho.
\]
Let
$F_1,\ldots,F_{d^2-1}$ be an orthonormal basis of the traceless Hermitian
matrices.  For the local coordinates
$\rho_\theta=\rho+\sum_b\theta_bF_b$, recall that the score of outcome $z$
in direction $F_b$ is
\[
  s_b(z)
  =
  \left.\frac{\partial}{\partial\theta_b}\log p_z(\theta)\right|_{\theta=0},
  \qquad
  p_z(\theta)
  =
  w_z\bra{v_z}\rho_\theta^{\otimes t}\ket{v_z}.
\]
\begin{samepage}
Differentiating and simplifying gives
\begin{equation}
\label{eq:overview-exact-score}
  s_b(z)
  =
  \Tr\bigl(F_b\rho^{-1/2}G_1(u_z)\rho^{-1/2}\bigr)
  =
  \Tr\bigl(F_b\rho^{-1/2}B_z\rho^{-1/2}\bigr).
\end{equation}
The second equality uses $G_1(u_z)=B_z+t\rho$ and $\Tr F_b=0$.  Thus the
score is determined by the deviation of $G_1(u_z)$ from $t\rho$.
\end{samepage}

The completeness relation for the measurement gives
\begin{align}
\label{eq:overview-exact-ensemble}
  \sum_zp_z\proj{u_z}
  &=
  \sum_z w_zR^{1/2}\proj{v_z}R^{1/2}
  =R,
  \notag\\
  \sum_zp_zG_1(u_z)
  &=t\rho.
\end{align}
The first equality says that the pairs $(p_z,u_z)$ form a pure-state
ensemble decomposition of $\rho^{\otimes t}$.  Applying the definition of
$G_1$ to both sides gives the second equality.  In particular, the average
of $G_1(u_z)$ over the measurement outcomes is $t\rho$, so
$B_z$ has mean zero.

Equation~\eqref{eq:overview-exact-score} expresses $s_b(z)$ as the inner
product of $F_b$ with
$\rho^{-1/2}B_z\rho^{-1/2}$.  Parseval's identity for
traceless Hermitian matrices
(\cref{fact:traceless-parseval}) and the Schatten norm comparison in
\cref{fact:schatten-comparisons} therefore give
\nopagebreak[4]
\begin{align}
\label{eq:overview-ensemble-score}
  \norm{s(z)}_2^2
  &=
  \sum_{b=1}^{d^2-1}
  \ip{F_b}{\rho^{-1/2}B_z\rho^{-1/2}}^2
  \notag\\
  &\leq
  \norm{\rho^{-1/2}B_z\rho^{-1/2}}_F^2
  \notag\\
  &\leq
  \frac{\norm{B_z}_F^2}
       {\lambda_{\min}(\rho)^2}.
\end{align}
By definition, the Fisher information matrix is
\[
  J_M(\rho)=\sum_zp_zs(z)s(z)^\mathsf T.
\]
Taking the trace and using
$\Tr(s(z)s(z)^\mathsf T)=\norm{s(z)}_2^2$ gives
\[
  \Tr J_M(\rho)=\sum_zp_z\norm{s(z)}_2^2.
\]
It remains to bound the average squared deviation
\[
  \sum_zp_z\norm{B_z}_F^2.
\]

The simplest bound for an individual outcome is not strong enough.
Since $G_1(u)$ is positive semidefinite with trace $t$,
\[
  \norm{G_1(u)}_F^2\leq t^2.
\]
Using only this estimate would give
$\Tr J_M(\rho)\lesssim t^2/\lambda_{\min}(\rho)^2$, which is of order
$d^2t^2$ when $\rho\succeq I_d/(4d)$.  This is too weak by a factor of
$\sqrt{t}$.  To obtain the stronger bound, we need to use the average
over the measurement outcomes.  The first equality in
\cref{eq:overview-exact-ensemble} determines the average of any quantity
that is linear in $\proj{u_z}$, whereas
$\norm{G_1(u_z)}_F^2$ is not linear in $\proj{u_z}$.  The Schur--Weyl
decomposition gives an upper bound on this squared norm that is linear
in $\proj{u_z}$ and can therefore be averaged using
\cref{eq:overview-exact-ensemble}.

\paragraph{Averaging over Schur subspaces.}
The space $(\C^d)^{\otimes t}$ decomposes into Schur subspaces indexed by
partitions $\lambda=(\lambda_1,\ldots,\lambda_d)$ of $t$, with
missing rows padded by zeros.  Let $\Pi_\lambda$ be the projector onto the
Schur subspace indexed by $\lambda$.  Weak Schur sampling on $\rho^{\otimes t}$ returns
$\lambda$ with probability $\Tr(\Pi_\lambda\rho^{\otimes t})$.  This is the
Schur--Weyl distribution $\SW^t(\operatorname{spec}\rho)$.  We write
$\ell(\lambda)$ for the number of nonzero parts.  The bound on the unfolded
matrix and the moment bound for the Schur--Weyl distribution
(\cref{fact:cll-unfolded-bound,fact:sw-second-moment}) give
\begin{equation}
\label{eq:overview-cll-projection-bound}
  \norm{G_1(u)}_F^2
  \leq
  \sum_{\substack{\lambda\vdash t\\\ell(\lambda)\leq d}}
  \norm{\Pi_\lambda u}_2^2\sum_{i=1}^d\lambda_i^2
\end{equation}
and
\begin{equation}
\label{eq:overview-cll-sw-moment-bound}
  \E_{\lambda\sim\SW^t(\operatorname{spec}\rho)}
  \Brac{\sum_{i=1}^d\lambda_i^2}
  \leq
  t(t-1)\Tr(\rho^2)+2t^{3/2}.
\end{equation}
The right-hand side of \cref{eq:overview-cll-projection-bound} is linear in
the squared projection norms $\norm{\Pi_\lambda u}_2^2$.  By
\cref{eq:overview-exact-ensemble},
$\sum_zp_z\norm{\Pi_\lambda u_z}_2^2
=\Tr(\Pi_\lambda\rho^{\otimes t})$.  Thus, averaging
\cref{eq:overview-cll-projection-bound} over $z$ replaces each projection
weight by the corresponding Schur--Weyl probability:
\begin{align}
\label{eq:overview-schur-weyl-step}
  \sum_zp_z\norm{G_1(u_z)}_F^2
  &\leq
  \sum_{\substack{\lambda\vdash t\\\ell(\lambda)\leq d}}
  \Tr(\Pi_\lambda\rho^{\otimes t})
  \sum_{i=1}^d\lambda_i^2
  \notag\\
  &=
  \E_{\lambda\sim\SW^t(\operatorname{spec}\rho)}
  \Brac{\sum_{i=1}^d\lambda_i^2}
  \notag\\
  &\leq
  t(t-1)\Tr(\rho^2)+2t^{3/2}.
\end{align}
Since the mean of $G_1(u_z)$ is $t\rho$, expanding the squared deviation
subtracts $t^2\Tr(\rho^2)$.  This cancels the
$t(t-1)\Tr(\rho^2)$ term in the preceding bound up to the nonpositive
remainder $-t\Tr(\rho^2)$:
\begin{equation}
\label{eq:overview-centered-unfolded-bound}
\begin{aligned}
  \sum_zp_z\norm{B_z}_F^2
  &=
  \sum_zp_z\norm{G_1(u_z)}_F^2-t^2\Tr(\rho^2)
  \\
  &\leq
  2t^{3/2}-t\Tr(\rho^2)
  \leq 2t^{3/2}.
\end{aligned}
\end{equation}
The improvement comes from averaging.  By
\cref{eq:overview-exact-ensemble}, the weighted average of the projectors
$\proj{u_z}$ is $\rho^{\otimes t}$, which determines the average weight in
each Schur subspace.  Subtracting the mean then removes the term of order
$t^2\Tr(\rho^2)$.

Averaging \cref{eq:overview-ensemble-score} over the measurement outcomes
and using \cref{eq:overview-centered-unfolded-bound} gives
\[
  \Tr J_M(\rho)
  \leq
  \frac{2t^{3/2}}{\lambda_{\min}(\rho)^2}.
\]
Together with the standard symmetric logarithmic derivative bound, this shows
that every POVM $M$ on $t$ copies and every full-rank state $\rho$ satisfy
\[
  \Tr J_M(\rho)
  \lesssim
  \min\Set{
    \frac{t^{3/2}}{\lambda_{\min}(\rho)^2},
    \frac{td^2}{\lambda_{\min}(\rho)}
  }.
\]
This is the block Fisher information bound in \cref{thm:block-fi}.
The two Schur--Weyl estimates used above come from \cite{CLL24}.  There they
are applied to an expansion of $(I_d/d+\Delta)^{\otimes t}$, which requires a
relation between $t$ and $\eps$.  Here the score formula and the average
in \cref{eq:overview-exact-ensemble} give the cancellation in
\cref{eq:overview-centered-unfolded-bound} directly, so the resulting bound
does not depend on $\eps$.

\subsection{The Prior and the Mutual Information Lower Bound}
\label{sec:overview-prior}

So far, we have proved the Fisher information bound for one measurement block
and shown how the adaptive chain rule turns it into a bound for the full
transcript, provided every state in the prior satisfies the required spectral
bound.  We now return to the prior in \cref{eq:overview-prior} and verify the
properties that complete the information comparison.

The lower spectral bound (\cref{eq:prior-spectral-bound}) gives
$\rho_\Delta\succeq I_d/(2d)$ throughout the support of the prior.  Thus
\cref{eq:overview-block-bound} applies to every prior draw and every block of
the adaptive protocol.  The log-Sobolev bound (\cref{eq:hard-prior-lsi}) then
converts the resulting Fisher information bound for the transcript into the
mutual information upper bound.  The lower bound on the probability of $K$
(\cref{eq:prior-conditioning-probability}) transfers the Gaussian estimate
before conditioning to the law $\mu$.  The resulting uniform bound on trace
norm balls (\cref{eq:prior-small-ball}) says that every ball of radius
$2\eps$ has prior mass exponentially small in $d^2$.  On success, the
estimator centers one such ball containing $\Delta$, so Fano's inequality
forces the transcript to contain order $d^2$ mutual information.  These
properties are stated together in \cref{prop:prior-properties} and proved in
\cref{sec:hard-prior}.

Truncation is needed because the unconditioned Gaussian can make
$I_d/d+\Delta$ nonpositive.  Since conditioning need not preserve a
log-Sobolev inequality, we use the convexity of $K$: the contraction theorem
in \cite{Caffarelli00,FGP20} gives the required bound.
$\Ent_\mu$ denotes entropy with respect to $\mu$.  For every nonnegative $f$ with
$\E_\mu f=1$,
\[
  \Ent_\mu(f)
  \leq
  \frac{\tau^2}{d^3}
  \E_\mu\Brac{\frac{\norm{\nabla f}_2^2}{f}}.
\]
Applying the mutual information comparison in \cref{eq:mi-from-fi} now gives
\[
  I(\Delta;Y)
  \leq
  \frac{\tau^2}{d^3}
  \E_{\Delta\sim\mu}\Tr J_Y(\rho_\Delta).
\]
This is the conversion from transcript Fisher information to the mutual
information upper bound used in \cref{sec:overview-information-upper}.

It remains to prove the lower bound on mutual information.  The needed
Gaussian estimates are the Hermitian norm estimates and the bound on the
lower tail of the trace norm
(\cref{lem:gaussian-hermitian-norms,lem:trace-norm-lower-tail}), proved in
\cref{sec:anti-concentration}.  Before conditioning, the perturbation has the
law of $Z/\sqrt\kappa$, where $Z$ is a standard Gaussian on
$\operatorname{Herm}_0(\C^d)$.  Its
trace norm satisfies
$\E\norm Z_1\gtrsim d^{3/2}$ and is $\sqrt d$-Lipschitz in Frobenius distance.
Gaussian concentration therefore makes the probability that $\norm Z_1$ is
below a sufficiently small constant multiple of $d^{3/2}$ exponentially
small in $d^2$.  The scale in \cref{eq:overview-prior} turns a trace norm
ball of radius $2\eps$ into a ball on this scale.  Because the estimator
determines the center of this ball, the estimate must be uniform over all
centers.  Anderson's inequality (\cref{fact:anderson}) provides this
uniformity, and $K$ has probability bounded below by an absolute constant
(\cref{eq:prior-conditioning-probability}).  Hence
\begin{equation}
\label{eq:overview-prior-ball}
  -\log\sup_{\Delta_0\in\operatorname{Herm}_0(\C^d)}
  \mu\Set{\Delta:\norm{\Delta-\Delta_0}_1\leq2\eps}
  \gtrsim d^2.
\end{equation}
Equivalently, there is an absolute constant $c>0$ such that every trace norm
ball of radius $2\eps$ has $\mu$-mass at most $e^{-cd^2}$, uniformly over its
center.  This is the uniform prior mass bound in
\cref{eq:prior-small-ball}.

Now suppose that $\widehat\rho$ succeeds, and set
$\widehat\Delta=\widehat\rho-I_d/d$.  Then
$\norm{\Delta-\widehat\Delta}_1\leq2\eps$.  Applying Fano's inequality for
metric balls (\cref{lem:small-ball-fano}) together with
\cref{eq:overview-prior-ball} gives
\[
  I(\Delta;Y)\gtrsim d^2,
\]
as stated in \cref{cor:information-lower-bound}.  This is the lower side of
the information comparison in \cref{eq:overview-information-sandwich}.

The condition $\eps\leq\eps_0$ in \cref{thm:main-lower-bound} comes from the
prior used in this proof.  Every state in its support lies within trace
distance $1/4$ of $I_d/d$, so extending the lower bound to larger constant
values of $\eps$ would require a different prior.

%% file: sections/related_work.tex
\section{Related Work}
\label{sec:related-work}

\paragraph{Unrestricted joint and single-copy tomography.}
Here an unrestricted joint measurement means a single POVM acting on all
$n$ copies; it is also called a collective measurement.  Haah, Harrow, Ji,
Wu, and Yu~\cite{HHJWWY16} proved that tomography requires
$n\gtrsim d^2/\eps^2$ copies even in this unrestricted model, while O'Donnell
and Wright~\cite{OW16} gave a matching procedure with
$n\lesssim d^2/\eps^2$.  At the opposite endpoint, in \cite{GKKT20},
Gu\c{t}\u{a}, Kahn, Kueng, and Tropp gave a single-copy procedure with
$n\lesssim d^3/\eps^2$.  In \cite{CHHLLS23}, Chen,
Huang, Li, Liu, and Sellke proved a matching lower bound and studied when
classical adaptivity can help state learning.  Lower bounds for restricted
single-copy measurements were also developed in \cite{LN25} by Lowe and
Nayak.  In \cite{ADLY25}, Acharya, Dharmavarapu, Liu, and Yu obtained
adaptive lower bounds for Pauli measurements and for measurements with a
prescribed number of outcomes.

\paragraph{Tomography with limited entanglement.}
In \cite{CLL24}, Chen, Li, and Liu studied tomography when each measurement
acts on at most $k$ copies and obtained the first tradeoff between the number
of copies measured at a time and the total copy complexity.
For some absolute constants $\alpha,\gamma>0$, their upper bound covered
$k\leq\min\{d^2,(\sqrt d/\eps)^\alpha\}$ up to logarithmic factors, while
their matching lower bound required $k\leq(1/\eps)^\gamma$.  The proof of
their lower bound expands $\rho^{\otimes k}$, and they conjectured that this
restriction on $k$ was an artifact.  Their lower bound follows the strategy of
\cite{CHHLLS23}: it controls how the posterior changes after each measurement
under a Gaussian perturbation centered at the maximally mixed state and
conditioned on an operator norm bound.  We use the same prior.  Instead of
expanding likelihood ratios and analyzing the posterior, we prove the uniform
Fisher information bound for a block measurement (\cref{thm:block-fi}), use
the prior's log-Sobolev inequality to bound mutual information, and apply
Fano's inequality for metric balls.
This removes the relation between $k$ and the accuracy.  Later, in
\cite{PSW26}, Pelecanos,
Spilecki, and Wright proved
the upper bound
$n\lesssim\max\{d^3/(\sqrt{k}\eps^2),d^2/\eps^2\}$ for every $k$, without
logarithmic factors.  Together with that upper bound, our main lower bound
(\cref{thm:main-lower-bound}) determines the copy complexity, up to absolute
constants, for every $k$ and all sufficiently small $\eps$.
Subsequently, in \cite{PSTW25}, Pelecanos, Spilecki, Tang, and Wright gave
another proof of this upper bound through a reduction from mixed state
tomography to pure state tomography.

\paragraph{Fisher information in state estimation.}
The symmetric logarithmic derivative gives a standard upper bound on the
classical Fisher information of every measurement
\cite{Helstrom76,Holevo11}.  In \cite{GillMassar00}, Gill and Massar studied
constraints on the classical Fisher information available in multiparameter
state estimation.  In \cite{ZhuHayashi18}, Zhu and Hayashi proved a sharp
two-copy Fisher information inequality at every state of full rank.
In \cite{CGZ26}, Chen, Gong, and Zhou obtained bounds on the sample complexity
of shadow estimation at high accuracy with adaptive measurements that act
jointly on a limited number of copies.  Their instance dependence is
expressed through inverse Fisher information.  The $d^2/\eps^2$
nonasymptotic minimax lower bound for unrestricted measurements follows from
\cite{HHJWWY16}.  Local Cram\'{e}r--Rao theory reflects the same scale, but it
does not capture the additional cost when each measurement acts on at most
$k$ copies.  Combined with the information argument, our Fisher information
bound for one block (\cref{thm:block-fi}) yields a copy lower bound larger than
the unrestricted $d^2/\eps^2$ rate by a factor $d/\sqrt{k}$ when $k\leq d^2$.
The bound applies at every full-rank state in terms of its smallest eigenvalue.

\paragraph{Schur--Weyl estimates.}
We use two estimates from
\cite[full version, Lemma~6.1 and the proof of Claim~3.27]{CLL24}: a bound on
the sum of the one-copy reduced states in terms of projections onto Schur
subspaces, and a moment bound for weak Schur sampling.  In
\cite[Supplemental Material, Section~VI, Eqs.~(S63)--(S66)]{ZhuHayashi18},
Zhu and Hayashi use the same change of variables for two copies.  We use its
direct extension to $t$ copies and center the sum of the reduced states at its
mean.  Combining this centered expression with the two estimates from
\cite{CLL24} gives a Fisher information bound that is uniform in $t$ and
depends only on the smallest eigenvalue of the state.

\paragraph{Information-theoretic lower bounds.}
The remaining information-theoretic ingredients come from prior work.  In
\cite[Lemma~1]{XuRaginsky17}, Xu and Raginsky give the mutual information
lower bound in terms of small-ball probabilities that we use.  In
\cite[Theorem~1]{ALPC19}, Aras, Lee, Pananjady, and Courtade bound mutual
information in terms of average Fisher information when a reference measure
satisfies a log-Sobolev inequality.  We use the specialization in which the
prior is also the reference measure.  In
\cite{Caffarelli00}, Caffarelli proved the contraction theorem for uniformly
log-concave measures.
In \cite{FGP20}, Fathi, Gozlan, and Prod'homme proved a version that permits
convex potentials to take the value $+\infty$.  This version gives the needed
log-Sobolev inequality for our truncated Gaussian.

%% file: sections/preliminaries.tex
\section{Preliminaries}
\label{sec:preliminaries}

This section fixes notation and recalls the results used in the proof.  The
main new ingredient for the lower bound is the Fisher information bound for a
block measurement (\cref{thm:block-fi}).

\subsection{Matrix and State Geometry}
\label{sec:prelim-matrix}

For a matrix $H$, we write $\norm H_1$, $\norm H_F$, and $\norm H_{\op}$
for its trace norm, Frobenius norm, and operator norm.  The inner product on
Hermitian matrices is $\ip HG=\Tr(HG)$.  The identity on $\C^d$ is $I_d$.

Let
\[
  \operatorname{Herm}_0(\C^d)
  =
  \Set{H\in\C^{d\times d}:H=H^\dagger,\ \Tr H=0}.
\]
This is a real Hilbert space of dimension $d^2-1$.  We fix a basis
$F_1,\ldots,F_{d^2-1}$ that is orthonormal under the Frobenius inner product
and identify $\theta\in\R^{d^2-1}$ with
$\sum_{b=1}^{d^2-1}\theta_bF_b\in\operatorname{Herm}_0(\C^d)$.  All gradients with respect to a
state parameter are taken in these coordinates.

The trace distance between states $\rho$ and $\sigma$ is
$\Dtr(\rho,\sigma)=\norm{\rho-\sigma}_1/2$.  We use natural logarithms, so
all information quantities are measured in nats.  For a set $E$, $\one_E$
denotes its indicator.

We repeatedly use the following two elementary calculations.  The first is
Parseval's identity, which says that the squared coordinates in an
orthonormal basis sum to the squared norm.

\begin{fact}[Parseval's identity for traceless Hermitian matrices]
\label{fact:traceless-parseval}
For every Hermitian matrix $H$,
\begin{equation}
\label{eq:traceless-parseval}
  \sum_{b=1}^{d^2-1}\ip{F_b}{H}^2
  =
  \norm{H-\frac{\Tr H}{d}I_d}_F^2
  =
  \norm H_F^2-\frac{(\Tr H)^2}{d}.
\end{equation}
In particular, the sum is at most $\norm H_F^2$.
\end{fact}

\begin{proof}
The matrices $F_1,\ldots,F_{d^2-1},I_d/\sqrt d$ form an orthonormal basis
of the real Hilbert space of Hermitian matrices.  Applying Parseval and
separating the coordinate
$\ip{I_d/\sqrt d}{H}=(\Tr H)/\sqrt d$ gives the last equality.  The matrix
$H-(\Tr H)I_d/d$ is the orthogonal projection of $H$ onto
$\operatorname{Herm}_0(\C^d)$, which gives the first equality.
\end{proof}

\begin{fact}[Schatten norm comparisons]
\label{fact:schatten-comparisons}
For compatible matrices $L,H,R$,
\[
  \norm{LHR}_F
  \leq
  \norm L_{\op}\norm H_F\norm R_{\op}.
\]
For Hermitian $H,G\in\C^{d\times d}$,
\begin{align}
  \norm H_F^2
  &\leq
  \norm H_{\op}\norm H_1,
  \label{eq:frobenius-trace-operator}\\
  \abs{\norm H_1-\norm G_1}
  &\leq
  \norm{H-G}_1
  \leq
  \sqrt d\norm{H-G}_F.
  \label{eq:trace-lipschitz}
\end{align}
\end{fact}

\begin{proof}
The first inequality follows from the variational definition of the
Frobenius norm.  The remaining inequalities follow from Hölder,
Cauchy--Schwarz for the singular values, and the reverse triangle inequality.
\end{proof}

\subsection{Measurements and Fisher Information}
\label{sec:prelim-fisher}

We next define the measurement model and collect the Fisher information facts
used in the block and adaptive arguments.

\paragraph{Measurements and adaptive protocols.}
A positive operator-valued measure (POVM) on a finite-dimensional Hilbert
space $\calH$ is a countably additive positive-operator-valued measure $M$
on a standard Borel space $(\calY,\mathcal F)$ such that
$M(\calY)=I_\calH$.  Measuring a state $R$ gives the probability measure
$P_R(B)=\Tr(M(B)R)$.  We write $\{M_y:y\in\calY\}$ and use sums for the
case in which the outcome set is finite.

Let $\xi$ denote the protocol's internal randomness, drawn independently of
the unknown state, and write
$Y_{<j}=(Y_1,\ldots,Y_{j-1})$.  Given $(\xi,Y_{<j})$, round $j$ selects
$t_j\in\{1,\ldots,k\}$ and a POVM
$M_j(\cdot\mid\xi,Y_{<j})$ on $(\C^d)^{\otimes t_j}$.  The stopping time $T$
satisfies $\sum_{j=1}^{T}t_j\leq n$ almost surely.  The transcript is
$Y=(\xi,Y_1,\ldots,Y_T)$, and only classical information is carried between
rounds.  We write $\E_\rho$ for expectation under the transcript law generated
by $\rho$.  We make the standard measurability assumption that the block
sizes and POVMs selected from each history form measurable kernels.

\paragraph{Classical Fisher information.}
Let $D=\dim\calH$.  Every POVM $M$ has a canonical probability measure
\[
  \nu_M(B)=\frac{\Tr M(B)}{D}.
\]
The Radon--Nikodym theorem, applied entrywise in any basis of $\calH$, gives a
measurable positive semidefinite operator $K(y)$ such that
\[
  M(\dd y)=K(y)\nu_M(\dd y).
\]
Moreover, $\Tr K(y)=D$ for $\nu_M$-almost every $y$.
For a differentiable family of states $R_\theta$, the outcome law has density
$p_\theta(y)=\Tr(K(y)R_\theta)$ with respect to $\nu_M$.  Its score and Fisher
information are
\[
  s_\theta(y)=\nabla_\theta\log p_\theta(y),
  \qquad
  J_Y(\theta)=\E_\theta[s_\theta s_\theta^\mathsf T],
\]
whenever $s_\theta\in L^2(P_\theta)$.  Thus
$\Tr J_Y(\theta)=\E_\theta\norm{s_\theta}_2^2$.  On a finite outcome space,
this is the familiar formula
\[
  J_Y(\theta)
  =
  \sum_{y:p_\theta(y)>0}
  p_\theta(y)s_\theta(y)s_\theta(y)^\mathsf T.
\]
For a full-rank state $\rho$, we restrict $\theta$ to a sufficiently small
open neighborhood of zero on which
$\rho_\theta=\rho+\sum_{b=1}^{d^2-1}\theta_bF_b$ remains a state.  For any
observation generated from this family, we abbreviate
$J_Y(\rho)=J_Y(0)$.  If $Y$ is the outcome of a POVM $M$ on $t$ copies, we
also write this matrix as $J_M(\rho)$.  For $F=\sum_bf_bF_b$ and
$f=(f_b)_b$, write $J_M(\rho)[F,F]=f^\mathsf T J_M(\rho)f$.

At a full-rank input, the score of every POVM is automatically in $L^2$.
Indeed, if $D_b$ is the derivative of the input state $R$ in coordinate $b$,
then $|s_b(y)|\leq\norm{R^{-1/2}D_bR^{-1/2}}_{\op}$ almost everywhere.  This
follows from $|\Tr(BA)|\leq\norm A_{\op}\Tr B$ for
$B=R^{1/2}K(y)R^{1/2}\succeq0$.

\begin{fact}[Fisher information for general outcomes]
\label{fact:fi-data-processing}
The score of the POVM outcome above has mean zero.  For every finite set
$\mathcal A$ and every fixed measurable map $\phi:\calY\to\mathcal A$ that
does not depend on $\theta$, the random variable $\phi(Y)$ takes one of
finitely many values determined by $Y$.  Its score satisfies
\[
  s_{\phi(Y)}(\phi(Y))
  =
  \E_\theta[s_Y(Y)\mid\phi(Y)].
\]
Consequently,
\begin{equation}
\label{eq:general-outcome-fi-supremum}
  \Tr J_Y(\theta)
  =
  \sup_{\mathcal A,\phi}\Tr J_{\phi(Y)}(\theta),
\end{equation}
where the supremum is over all such finite sets and maps.
More generally, suppose $\widetilde Y$ is the outcome of a POVM applied to a
family of states on a finite-dimensional space that is differentiable in
trace norm.  Let $Y$ be obtained from $\widetilde Y$ by a possibly randomized
rule that does not depend on $\theta$.  If the score of $\widetilde Y$ is in
$L^2$, then the score of $Y$ is in $L^2$ and
\[
  s_Y(Y)
  =
  \E_\theta[s_{\widetilde Y}(\widetilde Y)\mid Y],
\]
and hence
\[
  J_Y(\theta)
  \preceq
  J_{\widetilde Y}(\theta),
  \qquad
  \Tr J_Y(\theta)
  \leq
  \Tr J_{\widetilde Y}(\theta).
\]
Thus deleting part of an observation, or randomizing it without knowledge of
$\theta$, cannot increase its Fisher information.
\end{fact}

\begin{proof}
Since $\Tr K=D$ almost everywhere, differentiation may be passed through the
integral.  Differentiating $\int p_\theta\,\dd\nu_M=1$ gives
$\E_\theta s_\theta(Y)=0$.  If $A=\phi(Y)$ and
$B_a=\phi^{-1}(a)$, then
\[
  \Pr_\theta[A=a]
  =
  \int_{B_a}p_\theta(y)\nu_M(\dd y).
\]
Differentiating and dividing by $\Pr_\theta[A=a]$ proves the conditional
expectation identity.  Conditional Jensen gives
$J_A(\theta)\preceq J_Y(\theta)$ and hence the upper bound in
\cref{eq:general-outcome-fi-supremum}.

For the reverse bound, choose a simple vector-valued function $h$ with finite
range such that $\E_\theta\norm{s_Y-h}_2^2$ is arbitrarily small.  Set
$\mathcal A=\operatorname{range}(h)$ and $\phi=h$.  Conditional expectation is
the orthogonal projection in $L^2$, so
\[
  \E_\theta\norm{s_Y-\E_\theta[s_Y\mid\phi(Y)]}_2^2
  \leq
  \E_\theta\norm{s_Y-h}_2^2.
\]
The Pythagorean identity now shows that
$\Tr J_{\phi(Y)}(\theta)$ can be made arbitrarily close to
$\Tr J_Y(\theta)$.

For a parameter-independent randomized rule, let $Q$ be its postprocessing
kernel.  Differentiability of the state family in trace norm permits
differentiation of the POVM outcome law in total variation.  Applying $Q$ to
this signed derivative shows that the derivative of the law of $Y$ has density
$\E_\theta[s_{\widetilde Y}(\widetilde Y)\mid Y]$ with respect to the law of
$Y$.  This proves the displayed conditional expectation formula for the
score after postprocessing.  Conditional Jensen then gives
\[
  J_Y(\theta)
  =
  \E\Brac{\E[s_{\widetilde Y}\mid Y]
  \E[s_{\widetilde Y}\mid Y]^\mathsf T}
  \preceq
  \E[s_{\widetilde Y}s_{\widetilde Y}^\mathsf T].
\]
\end{proof}

An effect has rank one if it has the form $w\proj v$ for some $w>0$ and unit
vector $v$.  The vectors for different outcomes need not be orthogonal, so a
POVM with rank-one effects need not be a projective measurement.
\Cref{lem:block-score-ensemble} associates one such vector with each outcome.
These vectors both express the score and form a pure-state ensemble
decomposition of $\rho^{\otimes t}$.  The following standard reduction lets
us apply the lemma to an arbitrary finite POVM.

\begin{fact}[Splitting POVM effects into rank-one operators]
\label{fact:povm-effect-decomposition}
Let $M=\{M_y:y\in\mathcal Y\}$ be a finite POVM on a finite-dimensional
Hilbert space.  After discarding effects equal to zero, choose for each $y$
an orthonormal eigenbasis of the support of $M_y$ and write
\[
  M_y
  =
  \sum_{j=1}^{r_y}w_{y,j}\proj{v_{y,j}},
  \qquad
  w_{y,j} > 0,
\]
and define a POVM $\widetilde M$ with outcomes $(y,j)$ and effects
$\widetilde M_{y,j}=w_{y,j}\proj{v_{y,j}}$.  If $\widetilde Y=(Y,J)$
is its outcome, then discarding $J$ gives the outcome distribution of $M$.
Consequently, on every differentiable family of input states,
\[
  J_Y(\theta)\preceq J_{\widetilde Y}(\theta).
\]
In particular, an upper bound on Fisher information that holds for every
finite POVM with rank-one effects also holds for every finite POVM.
\end{fact}

\begin{proof}
The effects of $\widetilde M$ sum to the identity and have rank one.  For
every input state $\sigma_\theta$,
\[
  \sum_j\Pr_\theta[\widetilde Y=(y,j)]
  =\sum_j\Tr(\widetilde M_{y,j}\sigma_\theta)
  =\Tr(M_y\sigma_\theta).
\]
The right-hand side is the probability of outcome $y$ under $M$.  Thus the
map $(y,j)\mapsto y$ turns $\widetilde Y$ into $Y$, and data processing for
Fisher information (\cref{fact:fi-data-processing}) gives the claim.
\end{proof}

\begin{fact}[Fisher information chain rule for an adaptive transcript]
\label{fact:adaptive-fi-composition}
At every full-rank state $\rho$, the transcript of an adaptive protocol
satisfies
\begin{equation}
\label{eq:adaptive-fi-composition}
  \Tr J_Y(\rho)
  =
  \E_\rho\Brac{
    \sum_{j=1}^{T}
    \Tr J_{M_j(\cdot\mid\xi,Y_{<j})}(\rho)
  }.
\end{equation}
\end{fact}

\begin{proof}
Since every active round uses at least one copy, $T\leq n$.  Append
deterministic dummy outcomes after the protocol stops, so the transcript has
$n$ rounds; the dummy outcomes have zero score and contain no additional
information.  Given $(\xi,Y_{<j})$, the canonical trace measure of the chosen
POVM is a reference kernel for $Y_j$ that is independent of the state
parameter.  At an active
round, the score in coordinate $b$ satisfies
\[
  |s_{j,b}|
  \leq
  t_j\norm{\rho^{-1/2}F_b\rho^{-1/2}}_{\op}
  \leq
  k\norm{\rho^{-1/2}F_b\rho^{-1/2}}_{\op}.
\]
After shrinking the parameter neighborhood, the same bound holds up to a
factor of two uniformly in $\theta$, in the chosen POVM, and in the history.
Together with trace-norm differentiability of
$\rho_\theta^{\otimes t_j}$ and the measurability assumption above, this
justifies differentiating the finite product of conditional densities.  The
transcript therefore has score
\[
  s_Y=\sum_{j=1}^Ts_j,
\]
where $\E_\rho[s_j\mid\xi,Y_{<j}]=0$.  The increments are orthogonal in
$L^2$.  Expanding $\E_\rho\|s_Y\|_2^2$ and discarding the vanishing cross
terms gives \cref{eq:adaptive-fi-composition}.
\end{proof}

\paragraph{The symmetric logarithmic derivative.}
Let $\rho$ be full rank and let $F\in\operatorname{Herm}_0(\C^d)$.  The symmetric logarithmic
derivative $L_{\rho,F}$ is the unique Hermitian solution of
\begin{equation}
\label{eq:sld-definition}
  F=\frac12(\rho L_{\rho,F}+L_{\rho,F}\rho).
\end{equation}
The associated quantum Fisher information is
$Q_\rho(F)=\Tr(\rho L_{\rho,F}^2)$.

The coordinate formula and tensor additivity below are direct calculations.
The measurement inequality is due to Braunstein and Caves
\cite{BraunsteinCaves94}.

\begin{fact}[Braunstein--Caves inequality and tensor additivity]
\label{fact:quantum-fisher-information}
In an eigenbasis of $\rho$, with eigenvalues $r_1,\ldots,r_d$,
\begin{equation}
\label{eq:qfi-coordinate}
  Q_\rho(F)
  =
  2\sum_{i,j=1}^d\frac{|F_{ij}|^2}{r_i+r_j}.
\end{equation}
Every POVM $M$ on one copy satisfies
\begin{equation}
\label{eq:braunstein-caves}
  J_M(\rho)[F,F]\leq Q_\rho(F).
\end{equation}
For the scalar path $\rho_s=\rho+sF$, the quantity $Q_\rho(F)$ is additive:
\begin{equation}
\label{eq:qfi-additivity}
  Q_{\rho^{\otimes t}}\Paren{
    \left.\frac{\partial}{\partial s}
    \rho_s^{\otimes t}\right|_{s=0}
  }
  =
  tQ_\rho(F).
\end{equation}
\end{fact}

\begin{proof}
Equation~\eqref{eq:sld-definition} gives
$(L_{\rho,F})_{ij}=2F_{ij}/(r_i+r_j)$, which proves
\cref{eq:qfi-coordinate}.  For the path $s\mapsto\rho_s^{\otimes t}$, the
symmetric logarithmic derivative is the sum of the copies of $L_{\rho,F}$.  Its cross
terms vanish because $\Tr(\rho L_{\rho,F})=\Tr F=0$, proving
\cref{eq:qfi-additivity}.  The Braunstein--Caves measurement inequality
\cref{eq:braunstein-caves} is proved in \cite{BraunsteinCaves94}.
\end{proof}

The following direct consequence is the standard estimate used for large
measurement blocks.

\begin{corollary}[Fisher information bound from the symmetric logarithmic derivative]
\label{cor:qfi-cap}
Let $d\geq2$, let $\rho$ be a state of full rank, and let $t\geq1$.  Then
every POVM $M$ on $t$ copies satisfies
\[
  \Tr J_M(\rho)
  \leq
  \frac{t(d^2-1)}{\lambda_{\min}(\rho)}.
\]
\end{corollary}

\begin{proof}
Let $F\in\operatorname{Herm}_0(\C^d)$ satisfy $\norm F_F=1$, and let
$r_1,\ldots,r_d$ be the eigenvalues of $\rho$.  By the coordinate formula
\cref{eq:qfi-coordinate} and
$r_i+r_j\geq2\lambda_{\min}(\rho)$,
\[
  Q_\rho(F)
  =
  2\sum_{i,j}\frac{|F_{ij}|^2}{r_i+r_j}
  \leq\frac1{\lambda_{\min}(\rho)}.
\]
Tensor additivity and the measurement inequality give
$J_M(\rho)[F,F]\leq t/\lambda_{\min}(\rho)$.  Summing over the
$d^2-1$ basis directions proves the result.
\end{proof}

\subsection{Schur--Weyl Duality and the Unfolded Matrix}
\label{sec:prelim-schur-weyl}

Following the terminology in
\cite[full version, Section~2.2.2 and Definition~3.31]{CLL24}, we call
$G_1(u)$ the unfolded matrix of $u$.  In the proof of \cref{thm:block-fi}, this
matrix appears in the score of an arbitrary measurement.  Here we define it
and collect the representation-theoretic facts used later.

Let $t$ be a positive integer.  For an operator $H$ on $\C^d$ and
$1\leq\ell\leq t$, write
\[
  H^{(\ell)}
  =
  I_d^{\otimes(\ell-1)}\otimes H\otimes
  I_d^{\otimes(t-\ell)}.
\]
For a unit vector $u\in(\C^d)^{\otimes t}$, the partial trace
$\Tr_{\{1,\ldots,t\}\setminus\{\ell\}}(\proj u)$ is its reduced density
matrix on the $\ell$th copy: it is the unique operator whose expectation
against $H$ equals $\bra uH^{(\ell)}\ket u$.  After identifying each copy
with $\C^d$, define the unfolded matrix by
\begin{equation}
\label{eq:g1-definition}
  G_1(u)
  =
  \sum_{\ell=1}^t
  \Tr_{\{1,\ldots,t\}\setminus\{\ell\}}(\proj u).
\end{equation}
Equivalently, $G_1(u)$ is the sum of the one-copy reduced states.  Thus
$G_1(u)/t$ is their average.  The unnormalized matrix $G_1(u)$ is more
convenient in the formula for the score of each outcome in
\cref{eq:block-score-centered}.

\begin{fact}[Tensor derivatives and the unfolded matrix]
\label{fact:tensor-unfolded}
Let $\rho$ be full rank and $F\in\operatorname{Herm}_0(\C^d)$.  For $s$ in a neighborhood of
zero, let $\rho_s=\rho+sF$, let
$R=\rho^{\otimes t}$, and let $H=\rho^{-1/2}F\rho^{-1/2}$.  Then
\begin{equation}
\label{eq:tensor-derivative-square-root}
  R^{-1/2}
  \left.\frac{\partial}{\partial s}\rho_s^{\otimes t}\right|_{s=0}
  R^{-1/2}
  =
  \sum_{\ell=1}^tH^{(\ell)}.
\end{equation}
For every unit vector $u$,
\[
  G_1(u)\succeq0,
  \qquad
  \Tr G_1(u)=t,
  \qquad
  \bra u\sum_{\ell=1}^tH^{(\ell)}\ket u
  =
  \Tr(HG_1(u)).
\]
Moreover,
\[
  G_1(U^{\otimes t}u)=UG_1(u)U^\dagger
\]
for every unitary $U$ on $\C^d$.
\end{fact}

\begin{proof}
The tensor derivative follows from the product rule, and conjugation by
$R^{-1/2}$ gives \cref{eq:tensor-derivative-square-root}.  The remaining
claims follow directly from the definition of partial trace.
\end{proof}

We now give the form used in the proof of \cref{thm:block-fi}.  For effects
of rank one, the change of variables associates with each outcome a vector
that gives both its score and one term in a pure-state ensemble decomposition
of $\rho^{\otimes t}$.

\begin{lemma}[Scores of block measurement outcomes and a pure-state decomposition]
\label{lem:block-score-ensemble}
Let $\rho$ be a state of full rank, let $R=\rho^{\otimes t}$, and let
$M=\{w_z\proj{v_z}\}_z$ be a finite POVM whose effects have rank one, where
$\norm{v_z}_2=1$.  Define
\[
  q_z=\bra{v_z}R\ket{v_z},
  \qquad
  p_z=w_zq_z,
  \qquad
  \ket{u_z}=\frac{R^{1/2}\ket{v_z}}{\sqrt{q_z}},
  \qquad
  H_b=\rho^{-1/2}F_b\rho^{-1/2}.
\]
Then $p_z$ is the probability of outcome $z$, and its score satisfies
\begin{equation}
\label{eq:block-score-centered}
  s_b(z)
  =
  \Tr(H_bG_1(u_z))
  =
  \Tr\bigl(H_b(G_1(u_z)-t\rho)\bigr).
\end{equation}
The same vectors satisfy
\begin{equation}
\label{eq:block-ensemble-identities}
  \sum_zp_z\proj{u_z}=\rho^{\otimes t},
  \qquad
  \sum_zp_zG_1(u_z)=t\rho.
\end{equation}
\end{lemma}

\begin{proof}
The tensor product rule in \cref{fact:tensor-unfolded} gives
\[
  s_b(z)
  =
  \bra{u_z}\Paren{\sum_{\ell=1}^tH_b^{(\ell)}}\ket{u_z}
  =
  \Tr(H_bG_1(u_z)).
\]
Since $\Tr(H_b\rho)=\Tr F_b=0$, subtracting $t\rho$ does not change the
score.  Conjugating the completeness relation
$\sum_z w_z\proj{v_z}=I$ by $R^{1/2}$ gives
$\sum_zp_z\proj{u_z}=\rho^{\otimes t}$.  Taking the one-copy partial traces
and summing gives $\sum_zp_zG_1(u_z)=t\rho$.
\end{proof}

We next recall the part of Schur--Weyl duality needed here.  A partition
$\lambda\vdash t$ with at most $d$ parts is a nonincreasing tuple
\[
  \lambda=(\lambda_1,\ldots,\lambda_d),
  \qquad
  \lambda_1\geq\cdots\geq\lambda_d\geq0,
  \qquad
  \sum_{i=1}^d\lambda_i=t.
\]
Equivalently, $\lambda$ is a Young diagram with $t$ boxes and at most $d$
rows.  We write $\ell(\lambda)$ for the number of nonzero rows and always pad
the row lengths with zeros to obtain $d$ coordinates.

There are two natural actions on $(\C^d)^{\otimes t}$: the symmetric group
$S_t$ permutes the tensor factors, while $U(d)$ acts as $U^{\otimes t}$.
These actions commute.  For each $\lambda$, let $\mathsf{Sp}_\lambda$ and
$\mathsf V_\lambda^d$ be the vector spaces carrying the irreducible
representations of $S_t$ and $U(d)$, respectively, indexed by $\lambda$.
Schur--Weyl duality simultaneously decomposes the two actions as
\[
  (\C^d)^{\otimes t}
  =
  \bigoplus_{\substack{\lambda\vdash t\\\ell(\lambda)\leq d}}
  \mathsf{Sp}_\lambda\otimes\mathsf V_\lambda^d.
\]
Following \cite[full version, Definition~3.11]{CLL24}, we call
$\mathsf{Sp}_\lambda\otimes\mathsf V_\lambda^d$ the Schur subspace indexed by
$\lambda$ and write $\Pi_\lambda$ for its orthogonal projector.  The POVM
$\{\Pi_\lambda\}_\lambda$ is weak Schur sampling
\cite[full version, Definition~3.17]{CLL24}.  For a vector $u$,
$\norm{\Pi_\lambda u}_2^2$ is the squared norm of its projection onto the
Schur subspace indexed by $\lambda$.  For a state $\rho$,
$\Tr(\Pi_\lambda\rho^{\otimes t})$ is the probability of outcome $\lambda$.
For example, when $t=2$ and $d\geq2$,
the partitions
$(2,0,\ldots,0)$ and $(1,1,0,\ldots,0)$ index the symmetric and
antisymmetric subspaces.  If $\mathsf{Swap}(x\otimes y)=y\otimes x$, their
projectors are
\[
  \Pi_{(2)}=\frac{I_d^{\otimes2}+\mathsf{Swap}}2,
  \qquad
  \Pi_{(1,1)}=\frac{I_d^{\otimes2}-\mathsf{Swap}}2.
\]

For a state $\rho$, let $\operatorname{spec}\rho=(r_1,\ldots,r_d)$ denote
its eigenvalues in nonincreasing order, including zeros.  The outcome of weak
Schur sampling on $\rho^{\otimes t}$ has the Schur--Weyl distribution
\cite[full version, Definition~3.23 and Fact~3.24]{CLL24}:
\begin{equation}
\label{eq:sw-distribution}
  \Pr_{\lambda\sim\SW^t(\operatorname{spec}\rho)}[\lambda]
  =
  \Tr(\Pi_\lambda\rho^{\otimes t}).
\end{equation}
Because each $\Pi_\lambda$ commutes with $U^{\otimes t}$, these probabilities
depend only on the eigenvalues of $\rho$.  We denote the resulting
distribution for eigenvalue vector $r$ by $\SW^t(r)$.

The proof of the Fisher information bound for a block measurement
(\cref{thm:block-fi}) associates a pure state with each effect of rank one.
The next fact identifies the average of
$\norm{\Pi_\lambda u_z}_2^2$ over the resulting ensemble.

\begin{fact}[Schur projection weights for a pure-state decomposition]
\label{fact:ensemble-schur-projections}
Let $\{(p_z,u_z)\}_z$ be a finite or countable ensemble of unit vectors, so
$p_z\geq0$ and $\sum_zp_z=1$, and suppose
\[
  \sum_zp_z\proj{u_z}=\rho^{\otimes t}.
\]
Then, for every $\lambda\vdash t$ with $\ell(\lambda)\leq d$,
\begin{equation}
\label{eq:ensemble-schur-projections}
  \sum_zp_z\norm{\Pi_\lambda u_z}_2^2
  =
  \Tr(\Pi_\lambda\rho^{\otimes t})
  =
  \Pr_{\widetilde\lambda\sim\SW^t(\operatorname{spec}\rho)}
  [\widetilde\lambda=\lambda].
\end{equation}
More generally, if $\zeta$ is a probability measure on unit vectors and
$\int\proj u\,\dd\zeta(u)=\rho^{\otimes t}$, then
\[
  \int\norm{\Pi_\lambda u}_2^2\dd\zeta(u)
  =
  \Tr(\Pi_\lambda\rho^{\otimes t}).
\]
\end{fact}

\begin{proof}
Taking the trace of $\sum_zp_z\proj{u_z}=\rho^{\otimes t}$ against
$\Pi_\lambda$ and using the definition of the Schur--Weyl distribution in
\cref{eq:sw-distribution} gives the claim.  The proof with an integral is
identical.
\end{proof}

The following two estimates are from \cite{CLL24}.  The proof of the centered
average bound (\cref{lem:centered-average-g1}) applies them to the ensemble
associated with an arbitrary POVM.

\begin{fact}[Bound on the unfolded matrix
  {\cite[full version, Lemma~6.1]{CLL24}}]
\label{fact:cll-unfolded-bound}
For every unit vector $u\in(\C^d)^{\otimes t}$,
\begin{equation}
\label{eq:cll-unfolded-bound}
  \norm{G_1(u)}_F^2
  \leq
  \sum_{\substack{\lambda\vdash t\\\ell(\lambda)\leq d}}
  \norm{\Pi_\lambda u}_2^2
  \sum_{i=1}^d\lambda_i^2.
\end{equation}
\end{fact}

\begin{fact}[Moment bound for the Schur--Weyl distribution
  {\cite[full version, proof of Claim~3.27]{CLL24}}]
\label{fact:sw-second-moment}
Let $r=(r_1,\ldots,r_d)$ be a probability vector and let
$\lambda\sim\SW^t(r)$.  Then
\begin{equation}
\label{eq:sw-second-moment-bound}
  \E\Brac{\sum_{i=1}^d\lambda_i^2}
  \leq
  t(t-1)\sum_{i=1}^dr_i^2+2t^{3/2}.
\end{equation}
\end{fact}

The displayed statement of Claim~3.27 in \cite{CLL24} gives a coarser
bound.  Its proof uses the shifted power sum identity
\[
  \E\Brac{
    \sum_i\lambda_i^2-\sum_i(2i-1)\lambda_i}
  =
  t(t-1)\sum_i r_i^2
\]
and proves
$\E[\sum_i(2i-1)\lambda_i]\leq2t^{3/2}$, which gives
\cref{eq:sw-second-moment-bound}.  When $r=\operatorname{spec}\rho$, the
first term is $t(t-1)\Tr(\rho^2)$.  Centering the unfolded matrix later
subtracts $t^2\Tr(\rho^2)$, leaving $-t\Tr(\rho^2)\leq0$.

\subsection{Entropy, Mutual Information, and Testing}
\label{sec:prelim-information}

We collect the information identities used to pass from local Fisher
information to a minimax lower bound.

For probability measures $P$ and $Q$, their relative entropy is
\[
  D(P\Vert Q)
  =
  \int\log\Paren{\frac{\dd P}{\dd Q}}\dd P
\]
when $P$ is absolutely continuous with respect to $Q$, and is $+\infty$
otherwise.
The mutual information between random variables $X$ and $Y$ is
$I(X;Y)=D(P_{XY}\Vert P_X\otimes P_Y)$.

\begin{fact}[Mutual information for general outcomes {\cite{Gray11}}]
\label{fact:mutual-information-general-outcomes}
For every pair of jointly distributed random variables $X$ and $Y$,
\[
  I(X;Y)
  =
  \sup_{\mathcal A,\phi} I(X;\phi(Y)),
\]
where the supremum is over finite sets $\mathcal A$ and measurable maps $\phi$
from the outcome space of $Y$ to $\mathcal A$.  Thus a mutual information
upper bound follows if it holds for $\phi(Y)$ for every such map.
\end{fact}

For a probability measure $\mu$ and a nonnegative function $f$, define
\[
  \Ent_\mu(f)
  =
  \E_\mu[f\log f]
  -
  \E_\mu[f]\log\E_\mu[f].
\]

The next fact is the specialization to finite outcomes of
\cite[Theorem~1]{ALPC19} in which the prior also serves as the reference
distribution.  We include the short proof so that we need no regularity
assumption on a complete transcript.

\begin{fact}[Mutual information from a log-Sobolev inequality
  {\cite[Theorem~1]{ALPC19}}]
\label{fact:lsi-to-mutual-information}
Suppose $X\sim\mu$ takes values in a Euclidean space and, for every positive
smooth function $f$ with $\E_\mu f=1$,
\[
  \Ent_\mu(f)
  \leq
  C_\mu
  \E_\mu\Brac{\frac{\norm{\nabla f}_2^2}{f}}.
\]
Let $A$ take values in a finite set, set
$p_a(x)=\Pr[A=a\mid X=x]$ and $\bar p_a=\E_\mu p_a(X)$, and suppose $p_a$ is
positive and smooth whenever $\bar p_a>0$.  Then
\begin{equation}
\label{eq:lsi-to-mutual-information}
  I(X;A)
  \leq
  C_\mu
  \E_{X\sim\mu}\Tr J_A(X).
\end{equation}
\end{fact}

\begin{proof}
For each outcome with $\bar p_a>0$, set $f_a=p_a/\bar p_a$.  Outcomes with
$\bar p_a=0$ may be discarded.  Since
$\E_\mu f_a=1$,
\[
  I(X;A)
  =
  \sum_a\bar p_a\Ent_\mu(f_a).
\]
Applying the log-Sobolev inequality to each $f_a$ gives
\begin{align*}
  I(X;A)
  &\leq
  C_\mu\sum_a
  \int\frac{\norm{\nabla p_a(x)}_2^2}{p_a(x)}\,\dd\mu(x)\\
  &=
  C_\mu\E_{X\sim\mu}\Tr J_A(X),
\end{align*}
which proves \cref{eq:lsi-to-mutual-information}.
\end{proof}

The following metric form of Fano's inequality is the unconditional
specialization of \cite[Lemma~1]{XuRaginsky17}.  We use natural logarithms,
so its additive constant is $\log2$.

\begin{lemma}[Fano's inequality for metric balls]
\label{lem:small-ball-fano}
Let $(\mathsf X,\mathsf d)$ be a metric space, let $X\sim\pi$ take values in
$\mathsf X$, and let $\widehat X$ be an estimate of $X$ based on $Y$.  Let
$\eta\geq0$ and $0<\beta<1$.  If
\[
  \sup_{x_0\in\mathsf X}
  \pi\Set{x:\mathsf d(x,x_0)\leq\eta}
  \leq\beta
\]
and
$q=\Pr[\mathsf d(X,\widehat X(Y))\leq\eta]$, then
\begin{equation}
\label{eq:small-ball-fano}
  I(X;Y)
  \geq
  q\log\frac1\beta-\log2.
\end{equation}
\end{lemma}

\begin{proof}
Apply the cited inequality with its auxiliary variable constant, threshold
$1/2$, and loss
$\ell(x,x_0)=\one_{\{\mathsf d(x,x_0)>\eta\}}$.  Its small-ball term is at
most $\beta$ by hypothesis.  Rearranging and using
$I(X;\widehat X(Y))\leq I(X;Y)$ gives \cref{eq:small-ball-fano}.
\end{proof}

The Gaussian prior argument proves the main lower bound
(\cref{thm:main-lower-bound}) only above a fixed dimension.  The next result
handles the remaining dimensions.

\begin{fact}[Lower bound for unrestricted joint measurements
  {\cite[Theorem~3]{HHJWWY16}}]
\label{fact:hhj-collective-lower-bound}
Let $d\geq2$ and $0<\delta,\eta<1$.  Suppose a POVM on $n$ copies outputs a
state $\widehat\rho$ such that, for every $d$-dimensional state $\rho$,
\[
  \Pr\Brac{\Dtr(\rho,\widehat\rho)\leq\frac\delta2}
  \geq
  1-\eta.
\]
Then
\begin{equation}
\label{eq:hhj-collective-lower-bound}
  n
  \geq
  C_\eta\frac{d^2}{\delta^2}(1-\delta)^2
\end{equation}
holds.  The constant $C_\eta>0$ depends only on $\eta$.
\end{fact}

Every adaptive protocol on $n$ copies induces a POVM on all $n$ copies.
Thus the lower bound for unrestricted joint measurements
\cref{eq:hhj-collective-lower-bound} applies to the protocols considered
here.

\begin{corollary}[Lower bound for unrestricted joint measurements]
\label{cor:fixed-dimensional-patch}
For every $d\geq2$ and $0<\eps\leq1/4$, any tomography protocol using $n$
copies and succeeding with probability at least $2/3$ on every
$d$-dimensional state satisfies
\begin{equation}
\label{eq:fixed-dimensional-patch}
  n\gtrsim\frac{d^2}{\eps^2}.
\end{equation}
This holds even if a single joint POVM may act on all $n$ copies.
\end{corollary}

\begin{proof}
The claim follows from \cref{eq:hhj-collective-lower-bound} with $\eta=1/3$
and $\delta=2\eps$, since $1-2\eps\geq1/2$.
\end{proof}

\subsection{Gaussian Concentration and Log-Sobolev Inequalities}
\label{sec:prelim-gaussian}

We finish with the Gaussian inequalities used to construct and analyze the
prior.

The lower tail for the squared norm follows from
\cite[Lemma~1]{LaurentMassart00} with its parameter set to $m/16$.  The
Lipschitz bound is the Gaussian concentration inequality
\cite[Equation~(2.35)]{Ledoux01} applied to $-f$.

\begin{fact}[Gaussian norm and concentration]
\label{fact:gaussian-norm-concentration}
Let $Z$ be a standard Gaussian on an $m$-dimensional real Hilbert space.
Then $\norm Z_2^2$ has the $\chi_m^2$ distribution,
\[
  \E\norm Z_2^2=m,
  \qquad
  \Pr\Brac{\norm Z_2^2\leq\frac m2}
  \leq
  e^{-m/16}.
\]
If $f$ is $L$-Lipschitz, then, for every $u\geq0$,
\begin{equation}
\label{eq:gaussian-concentration}
  \Pr[f(Z)\leq\E f(Z)-u]
  \leq
  \exp\Paren{-\frac{u^2}{2L^2}}.
\end{equation}
\end{fact}

\begin{corollary}[Operator norm of a traceless Hermitian Gaussian]
\label{cor:hermitian-gaussian-op-tail}
Let $Z$ be a standard Gaussian on
$\operatorname{Herm}_0(\C^d)$ with the Frobenius inner product.  For all
sufficiently large $d$,
\begin{equation}
\label{eq:hermitian-gaussian-op-tail}
  \Pr\Brac{\norm Z_{\op}\geq8\sqrt d}
  \leq
  2\exp\Paren{-\bigl(8-2\log9\bigr)d}
  \leq\frac1{10}.
\end{equation}
\end{corollary}

\begin{proof}
Let $\mathcal N$ be a $1/4$-net of the unit sphere in $\C^d$, viewed as a
$2d$-dimensional real sphere, with $|\mathcal N|\leq9^{2d}$.  For every
Hermitian $H$,
\[
  \norm H_{\op}
  \leq
  2\max_{u\in\mathcal N}|u^\dagger Hu|.
\]
For fixed $u$, the random variable
$u^\dagger Zu=\langle Z,\proj u-I_d/d\rangle$ is centered Gaussian with
variance $1-1/d\leq1$.  Hence
\[
  \Pr\Brac{|u^\dagger Zu|\geq4\sqrt d}
  \leq2e^{-8d}.
\]
A union bound over $\mathcal N$ proves
\cref{eq:hermitian-gaussian-op-tail}.
\end{proof}

\begin{fact}[Anderson's inequality {\cite{Anderson55}}]
\label{fact:anderson}
Let $Z$ be a centered Gaussian on a finite-dimensional real vector space and
let $\mathcal C$ be a symmetric convex set.  Then, for every vector
$h$,
\[
  \Pr[Z-h\in\mathcal C]\leq\Pr[Z\in\mathcal C].
\]
\end{fact}

We next state the analytic inputs for the log-Sobolev estimate.  The first is
the Gaussian log-Sobolev inequality.

\begin{fact}[Gaussian log-Sobolev inequality
  {\cite{Gross75}}]
\label{fact:gaussian-lsi}
Let $\gamma_s$ be the standard Gaussian measure on $\R^s$.  Every smooth
$g$ satisfies
\[
  \Ent_{\gamma_s}(g^2)
  \leq
  2\E_{\gamma_s}\norm{\nabla g}_2^2.
\]
\end{fact}

The prior used below is a Gaussian conditioned on a convex set.  We use the
following form of Caffarelli's contraction theorem.

\begin{fact}[Caffarelli contraction theorem
  {\cite{Caffarelli00,FGP20}}]
\label{fact:caffarelli-convex-conditioning}
Let $\gamma_s$ be the standard Gaussian measure on $\R^s$, let
$K\subseteq\R^s$ be a closed convex set with nonempty interior, and let
$\gamma_s^K$ be $\gamma_s$ conditioned on $K$.  There is a $1$-Lipschitz map
$T:\R^s\to\R^s$ such that $T(X)\sim\gamma_s^K$ whenever $X\sim\gamma_s$.
\end{fact}

The required log-Sobolev inequality is a direct consequence.

\begin{corollary}[Log-Sobolev inequality under convex conditioning]
\label{cor:convex-gaussian-truncation}
Let $K\subseteq\R^s$ be a closed convex set with nonempty interior.  For
$\kappa>0$, let
\[
  \dd\mu_K(x)
  =
  \frac1{Z_K}
  e^{-\kappa\norm x_2^2/2}\one_K(x)\dd x.
\]
Here $Z_K$ is the normalizing constant.
Then every positive smooth function $f$ with $\E_{\mu_K}f=1$ satisfies
\begin{equation}
\label{eq:convex-truncation-lsi}
  \Ent_{\mu_K}(f)
  \leq
  \frac1{2\kappa}
  \E_{\mu_K}\Brac{\frac{\norm{\nabla f}_2^2}{f}}.
\end{equation}
\end{corollary}

\begin{proof}
For $\kappa=1$, the Caffarelli contraction theorem
(\cref{fact:caffarelli-convex-conditioning}) gives a $1$-Lipschitz map $T$
that sends the standard Gaussian measure to $\mu_K$.  If the right-hand side
below is infinite, there is nothing to prove.  Otherwise, the Sobolev chain
rule and $\norm{DT}_{\op}\leq1$ almost everywhere give, for
$g=\sqrt{f\circ T}$,
\[
  \norm{\nabla g(x)}_2^2
  \leq
  \frac{\norm{\nabla f(T(x))}_2^2}{4f(T(x))}
  \qquad\text{almost everywhere}.
\]
The Gaussian log-Sobolev inequality extends from smooth functions to
$W^{1,2}(\gamma_s)$ by truncation and smooth approximation.  Applying it to
$g$ therefore gives
\[
  \Ent_{\mu_K}(f)
  \leq
  \frac12\E_{\mu_K}
  \Brac{\frac{\norm{\nabla f}_2^2}{f}}.
\]
Rescaling gives \cref{eq:convex-truncation-lsi} for every $\kappa>0$.
\end{proof}

%% file: sections/main_lower_bound.tex
\section{Proof of the Main Lower Bound}
\label{sec:main-lower-bound-proof}

\mainlowerbound*

\begin{proof}[Proof of \cref{thm:main-lower-bound}]
Let $\Cprior$ be the absolute constant fixed in
\cref{prop:prior-properties}.  For all sufficiently large $d$ and
$0<\eps\leq1/(2\Cprior^2)$, the spectral bound for the prior
(\cref{eq:prior-spectral-bound}) shows that its support consists of states
to which \cref{prop:information-upper-bound} applies with $\rho_\star=I_d/d$.
Averaging the success guarantee and applying the mutual information lower
bound (\cref{cor:information-lower-bound}) gives
\[
  d^2
  \lesssim
  I(\Delta;Y)
  \lesssim
  \frac{\eps^2}{d}
  \E_{\Delta,Y}\Brac{
    \sum_{j=1}^{T}t_j\sqrt{\min\{t_j,d^2\}}
  }.
\]
Thus every successful protocol must satisfy
\begin{equation}
\label{eq:weighted-copy-requirement}
  \E_{\Delta,Y}\Brac{
    \sum_{j=1}^{T}t_j\sqrt{\min\{t_j,d^2\}}
  }
  \gtrsim
  \frac{d^3}{\eps^2}.
\end{equation}
Since $t_j\leq k$ and $\sum_{j=1}^{T}t_j\leq n$ almost surely,
\[
  \sum_{j=1}^{T}t_j\sqrt{\min\{t_j,d^2\}}
  \leq
  n\sqrt{\min\{k,d^2\}}.
\]
Combining this inequality with \cref{eq:weighted-copy-requirement} proves
\[
  n
  \gtrsim
  \frac{d^3}{\eps^2\sqrt{\min\{k,d^2\}}}
\]
for all sufficiently large $d$.

Let
\[
  \eps_0
  =
  \min\Set{
    \frac1{2\Cprior^2},
    \frac14
  }.
\]
The remaining dimensions form a fixed finite set.  For them, the lower bound
for unrestricted joint measurements (\cref{cor:fixed-dimensional-patch}) gives
\[
  n\gtrsim\frac{d^2}{\eps^2}
  \gtrsim
  \frac{d^3}{\eps^2\sqrt{\min\{k,d^2\}}}.
\]
Since only finitely many dimensions remain, this completes the proof.
\end{proof}

For the local corollary, translating the Gaussian prior proves the result in
all sufficiently large dimensions.  The unrestricted lower bound used above
is not local, so the finitely many remaining dimensions are handled by a test
between two states near $\rho_\star$.

\begin{proof}[Proof of \cref{cor:local-lower-bound}]
Set
\[
  L=\Cprior^2,
  \qquad
  \eps_1=\min\Set{\frac1{4\Cprior^2},\frac1{16}}.
\]
First consider the sufficiently large dimensions covered by
\cref{prop:prior-properties}.
For every $\Delta\in K$,
\[
  \norm\Delta_{\op}
  \leq
  \frac{\Cprior\tau}{d}
  =
  \frac{L\eps}{d}.
\]
Moreover, if $\rho_\star\succeq I_d/(2d)$ and $\eps\leq\eps_1$, then
\[
  \rho_\star+\Delta
  \succeq
  \frac{I_d}{2d}-\frac{L\eps}{d}I_d
  \succeq
  \frac{I_d}{4d}.
\]
Thus the translated prior is supported on the states in the local estimation
problem and satisfies the hypothesis of the mutual information upper bound
(\cref{prop:information-upper-bound}).  The mutual information lower bound
(\cref{cor:information-lower-bound}) applies unchanged to this translated
experiment centered at $\rho_\star$.  Combining these two bounds as in the
proof of \cref{thm:main-lower-bound} gives
\[
  n
  \gtrsim
  \frac{d^3}{\eps^2\sqrt{\min\{k,d^2\}}}
\]
for all sufficiently large $d$.

It remains to treat a fixed finite set of dimensions.  Work in an eigenbasis
of $\rho_\star$, let $m=\lfloor d/2\rfloor$, and set
\[
  H
  =
  \frac1d\operatorname{diag}(
    \underbrace{1,\ldots,1}_{m},
    \underbrace{-1,\ldots,-1}_{m},
    0,\ldots,0),
  \qquad
  \rho_\pm=\rho_\star\pm4\eps H.
\]
The final zeros are absent when $d$ is even.  Since $\eps\leq1/16$, both
states have smallest eigenvalue at least $1/(4d)$, and
$\norm{\rho_\pm-\rho_\star}_{\op}=4\eps/d\leq L\eps/d$.  They also satisfy
\[
  \Dtr(\rho_+,\rho_-)
  =
  \frac12\norm{8\eps H}_1
  =
  \frac{8m}{d}\eps
  \geq
  \frac83\eps
  >2\eps.
\]

Let $p_+$ and $p_-$ be their eigenvalue distributions in this basis.  Every
entry of $p_-$ is at least $1/(4d)$, while the two distributions differ by
$8\eps/d$ in at most $2m$ coordinates.  The elementary bound of relative
entropy by chi-squared divergence therefore gives
\begin{equation}
\label{eq:local-two-state-relative-entropy}
  D(p_+\Vert p_-)
  \leq
  \sum_i\frac{(p_{+,i}-p_{-,i})^2}{p_{-,i}}
  \leq
  256\eps^2.
\end{equation}
Because the states commute, the transcript of any protocol on at most $n$
copies is a postprocessing of $n$ independent samples from their common
eigenbasis.  Let $P_+$ and $P_-$ be the two transcript laws.  Data processing
and additivity of relative entropy give
\[
  D(P_+\Vert P_-)
  \leq
  nD(p_+\Vert p_-)
  \leq
  256n\eps^2.
\]
The two trace distance balls of radius $\eps$ are disjoint.  Hence an estimator that succeeds with
probability at least $2/3$ under each state distinguishes $P_+$ from $P_-$
with success probability at least $2/3$.  It follows that
$\operatorname{TV}(P_+,P_-)\geq1/3$.  Pinsker's inequality now gives
\[
  \frac29
  \leq
  D(P_+\Vert P_-)
  \leq
  256n\eps^2.
\]
Thus $n\geq1/(1152\eps^2)$.  Since the remaining dimensions form a fixed
finite set and $\sqrt{\min\{k,d^2\}}\geq1$, this implies
\[
  n
  \gtrsim
  \frac{d^3}{\eps^2\sqrt{\min\{k,d^2\}}},
\]
uniformly over $k$.
\end{proof}

We finally remove the almost sure bound on the number of copies.

\begin{proof}[Proof of \cref{cor:expected-copy-budget}]
Let $N$ be the number of copies used by the original protocol.  Define a new
protocol that follows it but stops before starting the first block that would
bring the total number of copies above $6n$.  If it stops early, it outputs
any fixed state.  For every input state, Markov's inequality gives
\[
  \Pr[N>6n]\leq\frac16.
\]
The new protocol therefore uses at most $6n$ copies and succeeds with
probability at least $1/2$ on every state.

For all sufficiently large dimensions, apply
\cref{cor:information-lower-bound,prop:information-upper-bound} to the new
protocol.  The first result allows success probability $1/2$, and the second
applies with copy budget $6n$.  Their comparison gives
\[
  n
  \gtrsim
  \frac{d^3}{\eps^2\sqrt{\min\{k,d^2\}}}
\]
for both estimation problems.  For the fixed dimensions in the main theorem,
the lower bound for unrestricted joint measurements in
\cref{fact:hhj-collective-lower-bound}, applied with $\eta=1/2$ and
$\delta=2\eps$ to the same truncated protocol, gives
$n\gtrsim d^2/\eps^2$, which implies the displayed rate because only finitely
many dimensions remain.

For the fixed dimensions in the local statement, use the two states
$\rho_+$ and $\rho_-$ constructed in the proof of
\cref{cor:local-lower-bound}, and truncate instead at $12n$.  The resulting
protocol uses at most $12n$ copies and succeeds under
each state with probability at least $7/12$.  It therefore distinguishes the
two transcript laws with success probability at least $7/12$, so their total
variation distance is at least $1/6$.  Repeating the relative entropy
calculation in \cref{eq:local-two-state-relative-entropy} and applying
Pinsker's inequality gives
\[
  \frac1{18}
  \leq
  D(P_+\Vert P_-)
  \leq
  12n\cdot256\eps^2.
\]
Hence $n\gtrsim1/\eps^2$.  Because these dimensions form a fixed finite set,
this also implies the local rate
\[
  n
  \gtrsim
  \frac{d^3}{\eps^2\sqrt{\min\{k,d^2\}}}.
\]
\end{proof}

%% file: sections/anti_concentration.tex
\section{The Prior and the Mutual Information Lower Bound}
\label{sec:anti-concentration}

\subsection{The Truncated Gaussian Prior}
\label{sec:hard-prior}

Let $\Cprior\geq2$ be an absolute constant, to be fixed below, and let
$\tau=\Cprior\eps$.  Set
\begin{equation}
\label{eq:prior-parameters}
  \kappa=\frac{d^3}{2\tau^2}
\end{equation}
and
\begin{equation}
\label{eq:prior-support}
  K
  =
  \Set{
    \Delta\in\operatorname{Herm}_0(\C^d):
    \norm\Delta_{\op}\leq\frac{\Cprior\tau}{d}
  }.
\end{equation}
Let $\bar\mu$ be the centered Gaussian measure on
$\operatorname{Herm}_0(\C^d)$ with covariance $\kappa^{-1}I$, and let $\mu$
be its conditional law given $K$:
\begin{equation}
\label{eq:hard-prior}
  \dd\mu(\Delta)
  \propto
  \exp\Paren{-\frac\kappa2\norm\Delta_F^2}
  \one_K(\Delta)\dd\Delta.
\end{equation}
We use this distribution of perturbations around a fixed state $\rho_\star$;
the corresponding input is $\rho_\star+\Delta$.  For the main theorem,
$\rho_\star=I_d/d$.

The next proposition collects the properties of this prior used in the
information bounds.

\begin{proposition}[Properties of the truncated Gaussian prior]
\label{prop:prior-properties}
There is an absolute choice of $\Cprior\geq2$ such that, for all sufficiently
large $d$, the following statements hold.
\begin{enumerate}
\item The conditioning event has probability
  \begin{equation}
  \label{eq:prior-conditioning-probability}
    \bar\mu(K)\geq\frac12.
  \end{equation}
\item If $\eps\leq1/(2\Cprior^2)$, then every $\Delta\in K$ defines a state
  and
  \begin{equation}
  \label{eq:prior-spectral-bound}
    \frac{I_d}{d}+\Delta\succeq\frac{I_d}{2d}.
  \end{equation}
\item For every positive smooth function $f$ with $\E_\mu f=1$,
  \begin{equation}
  \label{eq:hard-prior-lsi}
    \Ent_\mu(f)
    \leq
    \frac{1}{2\kappa}
    \E_\mu\Brac{\frac{\norm{\nabla f}_2^2}{f}}
    =
    \frac{\tau^2}{d^3}
    \E_\mu\Brac{\frac{\norm{\nabla f}_2^2}{f}}.
  \end{equation}
\item Uniformly over $\Delta_0\in\operatorname{Herm}_0(\C^d)$,
  \begin{equation}
  \label{eq:prior-small-ball}
    -\log
    \mu\Set{\Delta:\norm{\Delta-\Delta_0}_1\leq2\eps}
    \gtrsim d^2.
  \end{equation}
\end{enumerate}
\end{proposition}

We fix this value of $\Cprior$ for the rest of the paper.  The last property,
together with Fano's inequality for metric balls, gives the information lower
bound used in the main proof.

\begin{corollary}[Mutual information required for tomography]
\label{cor:information-lower-bound}
For all sufficiently large $d$, let $\rho_\star$ be a fixed state such that
$\rho_\star+\Delta$ is a state for every $\Delta\in K$.  Let $\Delta\sim\mu$,
and let $Y$ be the transcript generated by the input $\rho_\star+\Delta$.  If,
for some $1/2\leq q\leq1$, an estimator $\widehat\rho(Y)$ satisfies
\[
  \Pr\Brac{
    \Dtr(\rho_\star+\Delta,\widehat\rho(Y))\leq\eps
  }
  \geq q,
\]
then
\begin{equation}
\label{eq:information-lower-bound}
  I(\Delta;Y)\gtrsim d^2.
\end{equation}
\end{corollary}

\begin{proof}
Let $\widehat\Delta=\widehat\rho-\rho_\star$.  Then
\[
  \Dtr(\rho_\star+\Delta,\widehat\rho)
  =\frac12\norm{\Delta-\widehat\Delta}_1.
\]
Let $\beta$ be the supremum of the prior masses of trace norm balls of radius
$2\eps$.  By \cref{eq:prior-small-ball},
$\log(1/\beta)\gtrsim d^2$.  Fano's inequality for metric balls
(\cref{lem:small-ball-fano}) therefore gives
\[
  I(\Delta;Y)
  \geq q\log\frac1\beta-\log2
  \gtrsim d^2.
\]
\end{proof}

It remains to prove \cref{prop:prior-properties}.  Write
$\Delta=Z/\sqrt\kappa$ before conditioning, where $Z$ is a standard Gaussian
on $\operatorname{Herm}_0(\C^d)$.  We use the following two consequences of
Gaussian concentration.

\begin{lemma}[Gaussian Hermitian norm estimates]
\label{lem:gaussian-hermitian-norms}
There are absolute constants $C_0,c_0>0$ and $d_0\in\N$ such that, for every
$d\geq d_0$ and every standard Gaussian $Z$ on
$\operatorname{Herm}_0(\C^d)$,
\begin{align}
  \Pr\Brac{\norm Z_{\op}\leq C_0\sqrt d}&\geq\frac{9}{10},
  \label{eq:gaussian-op-event}\\
  \Pr\Brac{\norm Z_F^2\geq\frac{d^2-1}{2}}&\geq\frac{9}{10},
  \label{eq:gaussian-frob-event}\\
  \E\norm Z_1&\geq c_0d^{3/2}.
  \label{eq:gaussian-trace-mean}
\end{align}
\end{lemma}

\begin{lemma}[Lower tail of the Gaussian trace norm]
\label{lem:trace-norm-lower-tail}
Let $c_0,d_0$ be as in \cref{lem:gaussian-hermitian-norms}.  There is an
absolute constant $a_0>0$ such that, for every $d\geq d_0$ and every standard
Gaussian $Z$ on $\operatorname{Herm}_0(\C^d)$,
\begin{equation}
\label{eq:trace-norm-lower-tail}
  \Pr\Brac{
    \norm Z_1\leq\frac{c_0}{2}d^{3/2}
  }
  \leq e^{-a_0d^2}.
\end{equation}
\end{lemma}

\begin{proof}[Proof of \cref{prop:prior-properties}]
Choose $\Cprior\geq2$ once so that
\[
  \frac{\Cprior}{\sqrt2}\geq C_0,
  \qquad
  \frac{\sqrt2}{\Cprior}\leq\frac{c_0}{2}.
\]
The event $\Delta\in K$ is
\[
  \norm Z_{\op}
  \leq
  \sqrt\kappa\frac{\Cprior\tau}{d}
  =
  \frac{\Cprior}{\sqrt2}\sqrt d.
\]
Thus \cref{eq:gaussian-op-event} gives
$\bar\mu(K)\geq9/10$, which proves \cref{eq:prior-conditioning-probability}.

If $\eps\leq1/(2\Cprior^2)$, then
$\|\Delta\|_{\op}\leq\Cprior^2\eps/d\leq1/(2d)$ on $K$.
Since $\Tr\Delta=0$, the matrix $I_d/d+\Delta$ has trace one, while
\[
  \frac{I_d}{d}+\Delta
  \succeq
  \Paren{\frac1d-\norm\Delta_{\op}}I_d
  \succeq
  \frac{I_d}{2d}.
\]
Thus it is a state and satisfies \cref{eq:prior-spectral-bound}.

The set $K$ is closed and convex and has nonempty interior.  The log-Sobolev
inequality under convex conditioning
(\cref{cor:convex-gaussian-truncation}), with
$\kappa=d^3/(2\tau^2)$, gives \cref{eq:hard-prior-lsi}.

It remains to prove the uniform bound on trace norm balls.  Since
$\tau=\Cprior\eps$,
\begin{equation}
\label{eq:scaled-small-ball-radius}
  2\eps\sqrt\kappa
  =
  \frac{\sqrt2}{\Cprior}d^{3/2}
  \leq
  \frac{c_0}{2}d^{3/2}.
\end{equation}
The set $\{H:\norm H_1\leq\eta\}$ is symmetric and convex.  Anderson's
inequality (\cref{fact:anderson}) and
\cref{lem:trace-norm-lower-tail} therefore give, for every $\Delta_0$,
\begin{align*}
  \bar\mu\Set{\norm{\Delta-\Delta_0}_1\leq2\eps}
  &=
  \Pr\Brac{\norm{Z-\sqrt\kappa\Delta_0}_1
    \leq2\eps\sqrt\kappa}\\
  &\leq
  \Pr\Brac{\norm Z_1\leq2\eps\sqrt\kappa}\\
  &\leq e^{-a_0d^2}.
\end{align*}
Conditioning on $K$ and using $\bar\mu(K)\geq1/2$ now yields
\[
  \mu\Set{\norm{\Delta-\Delta_0}_1\leq2\eps}
  \leq2e^{-a_0d^2}.
\]
Taking the negative logarithm proves \cref{eq:prior-small-ball}.
\end{proof}

\begin{proof}[Proof of \cref{lem:gaussian-hermitian-norms}]
The elementary net bound in \cref{cor:hermitian-gaussian-op-tail} gives
\cref{eq:gaussian-op-event} with $C_0=8$ after increasing $d_0$ if needed.

For the Frobenius norm estimate, the chi-square lower tail bound in
\cref{fact:gaussian-norm-concentration}, applied to the
$(d^2-1)$-dimensional Gaussian space $\operatorname{Herm}_0(\C^d)$, gives
\[
  \Pr\Brac{
    \norm Z_F^2<\frac{d^2-1}{2}
  }
  \leq e^{-(d^2-1)/16},
\]
which is at most $1/10$ for all sufficiently large $d$.

The events in \cref{eq:gaussian-op-event,eq:gaussian-frob-event} intersect
with probability at least $4/5$.  On their intersection, the Schatten norm
comparison in \cref{fact:schatten-comparisons} gives
\[
  \norm Z_1
  \geq
  \frac{\norm Z_F^2}{\norm Z_{\op}}
  \geq
  \frac{d^2-1}{2C_0\sqrt d}
  \gtrsim d^{3/2}.
\]
Taking expectations proves \cref{eq:gaussian-trace-mean} for an absolute
$c_0>0$ after increasing $d_0$ if needed.
\end{proof}

\begin{proof}[Proof of \cref{lem:trace-norm-lower-tail}]
By \cref{fact:schatten-comparisons}, the function $f(Z)=\norm Z_1$ is
$\sqrt d$-Lipschitz with respect to Frobenius distance.  Apply the Gaussian
concentration bound in \cref{fact:gaussian-norm-concentration} with deviation
$(c_0/2)d^{3/2}$ and use \cref{eq:gaussian-trace-mean}.  For an absolute
constant $a_0>0$, this gives
\[
  \Pr\Brac{
    \norm Z_1\leq\frac{c_0}{2}d^{3/2}
  }
  \leq e^{-a_0d^2}.
\]
\end{proof}

%% file: sections/global_information.tex
\section{Adaptive Protocols and Mutual Information}
\label{sec:global-information}

This section derives the information upper bound for an adaptive protocol
from the Fisher information bound for one block.  Fisher information is useful
here because it has a chain rule at each fixed state.  By contrast, the
conditional law of $\Delta$ changes after every outcome, so a mutual
information bound for one block under the original prior cannot simply be
reused after conditioning on the transcript.

\begin{proposition}[Information bound for an adaptive protocol]
\label{prop:information-upper-bound}
Let $\rho\succeq I_d/(4d)$, and let $Y$ be the transcript of a protocol in
\cref{def:block-protocol-intro} that uses at most $n$ copies almost surely.
Then
\begin{align}
\label{eq:adaptive-block-fi}
  \Tr J_Y(\rho)
  \lesssim
  d^2\E_\rho\Brac{
    \sum_{j=1}^{T}t_j\sqrt{\min\{t_j,d^2\}}
  }
  \leq
  d^2n\sqrt{\min\{k,d^2\}}.
\end{align}

Now let $\rho_\star$ be a fixed state such that
\[
  \rho_\star+\Delta\succeq\frac{I_d}{4d}
  \qquad\text{for every $\Delta\in K$},
\]
let $\Delta\sim\mu$, and run the protocol on $\rho_\star+\Delta$.  Then
\begin{equation}
\label{eq:mi-from-fi}
  I(\Delta;Y)
  \leq
  \frac{\tau^2}{d^3}
  \E_{\Delta\sim\mu}\Tr J_Y(\rho_\star+\Delta),
\end{equation}
and consequently
\begin{align}
\label{eq:information-upper-bound}
  I(\Delta;Y)
  \lesssim
  \frac{\tau^2}{d}
  \E_{\Delta,Y}\Brac{
    \sum_{j=1}^{T}t_j\sqrt{\min\{t_j,d^2\}}
  }
  &\leq
  \frac{\tau^2n}{d}\sqrt{\min\{k,d^2\}}.
\end{align}
\end{proposition}

\begin{proof}
At a fixed state $\rho$, the Fisher information chain rule for an adaptive
transcript (\cref{fact:adaptive-fi-composition}) and the block bound
(\cref{thm:block-fi}) give
\begin{align*}
  \Tr J_Y(\rho)
  &\lesssim
  d^2\E_\rho\Brac{
    \sum_{j=1}^{T}t_j\sqrt{\min\{t_j,d^2\}}
  }\\
  &\leq
  d^2\sqrt{\min\{k,d^2\}}
  \E_\rho\Brac{\sum_{j=1}^{T}t_j}\\
  &\leq d^2n\sqrt{\min\{k,d^2\}},
\end{align*}
which proves \cref{eq:adaptive-block-fi}.

For the mutual information comparison, let $\phi$ be a measurable function of
the transcript with finite range, and set $A=\phi(Y)$.  The protocol and
$\phi$ induce a finite POVM $\{E_a\}$ on $n$ copies; a branch that stops early
acts as the identity on its unused copies.  The outcome probabilities are
\[
  p_a(\Delta)
  =
  \Pr[A=a\mid\Delta]
  =
  \Tr\Paren{E_a(\rho_\star+\Delta)^{\otimes n}}.
\]
These probabilities are smooth in $\Delta$.  After zero effects are discarded,
they are positive because every
$\rho_\star+\Delta$ has full rank.  The mutual information bound from a
log-Sobolev inequality (\cref{fact:lsi-to-mutual-information}) and
\cref{eq:hard-prior-lsi} therefore give
\[
  I(\Delta;A)
  \leq
  \frac{\tau^2}{d^3}
  \E_{\Delta\sim\mu}\Tr J_A(\rho_\star+\Delta)
  \leq
  \frac{\tau^2}{d^3}
  \E_{\Delta\sim\mu}\Tr J_Y(\rho_\star+\Delta).
\]
The second inequality is Fisher information data processing
(\cref{fact:fi-data-processing}).  Taking the supremum over $\phi$ as in
\cref{fact:mutual-information-general-outcomes} proves
\cref{eq:mi-from-fi}.  Finally, average \cref{eq:adaptive-block-fi} over
$\Delta$ and substitute it into \cref{eq:mi-from-fi}; this gives the first
bound in \cref{eq:information-upper-bound}.  The second follows from
$t_j\leq k$ and $\sum_{j=1}^{T}t_j\leq n$.
\end{proof}

%% file: sections/block_fisher_information.tex
\section{Fisher Information of a Block Measurement}
\label{sec:block-fi}

\begin{theorem}[Fisher information bound for a block measurement]
\label{thm:block-fi}
Let $d\geq2$, let $t\geq1$, and let $\rho$ be a state of full rank.  Every
POVM $M$ on $t$ copies satisfies
\begin{equation}
\label{eq:block-fi-eigenvalue}
  \Tr J_M(\rho)
  \leq
  \min\Set{
    \frac{2t^{3/2}}{\lambda_{\min}(\rho)^2},
    \frac{t(d^2-1)}{\lambda_{\min}(\rho)}
  }.
\end{equation}
In particular, if
\begin{equation}
\label{eq:block-conditioning}
  \rho\succeq\frac{I_d}{4d},
\end{equation}
then
\begin{equation}
\label{eq:block-fi-main}
  \Tr J_M(\rho)
  \lesssim
  \min\Set{d^2t^{3/2},d^3t}
  =
  d^2t\sqrt{\min\{t,d^2\}}.
\end{equation}
\end{theorem}

The second bound in \cref{eq:block-fi-eigenvalue} is the standard consequence
of the symmetric logarithmic derivative in \cref{cor:qfi-cap}.  The first
bound for arbitrary $t$ and arbitrary $\rho$ is the new ingredient.  When
$\lambda_{\min}(\rho)$ is of order $1/d$, the first bound improves the
standard estimate for $t\leq d^2$ and yields the dependence on the block size
needed for the main lower bound.

The proof uses the following estimate.  The Schur--Weyl moment bound contains
the term $t(t-1)\Tr(\rho^2)$.  Centering $G_1(u_z)$ at its mean $t\rho$
subtracts $t^2\Tr(\rho^2)$, so the two contributions involving
$\Tr(\rho^2)$ sum to $-t\Tr(\rho^2)\leq0$.

\begin{lemma}[Average squared deviation of the unfolded matrix]
\label{lem:centered-average-g1}
Let $p_z$ and $u_z$ be the outcome probabilities and vectors from
\cref{lem:block-score-ensemble}.  Then
\begin{equation}
\label{eq:centered-average-g1}
  \sum_zp_z\norm{G_1(u_z)-t\rho}_F^2
  \leq
  2t^{3/2}.
\end{equation}
\end{lemma}

We first show how this lemma proves the theorem.

\begin{proof}[Proof of \cref{thm:block-fi}]
The symmetric logarithmic derivative bound in \cref{cor:qfi-cap} gives
\[
  \Tr J_M(\rho)
  \leq
  \frac{t(d^2-1)}{\lambda_{\min}(\rho)}.
\]
For the first term in the upper bound, first suppose that $M$ is finite.  Splitting its effects
into rank-one operators as in \cref{fact:povm-effect-decomposition} can only
increase Fisher information, so it is enough to consider effects of the form
$w_z\proj{v_z}$.  By the formula for the score of each outcome in
\cref{eq:block-score-centered}, Parseval's identity for traceless Hermitian
matrices (\cref{fact:traceless-parseval}), and the Schatten norm comparison
in \cref{fact:schatten-comparisons},
\begin{align}
  \sum_{b=1}^{d^2-1}s_b(z)^2
  &\leq
  \norm{
    \rho^{-1/2}\bigl(G_1(u_z)-t\rho\bigr)\rho^{-1/2}
  }_F^2
  \notag\\
  &\leq
  \frac1{\lambda_{\min}(\rho)^2}
  \norm{G_1(u_z)-t\rho}_F^2.
\label{eq:centered-score-parseval}
\end{align}
Averaging over the outcomes and applying
\cref{lem:centered-average-g1} gives
\[
  \Tr J_M(\rho)
  \leq
  \frac{2t^{3/2}}{\lambda_{\min}(\rho)^2}.
\]
Now let $M$ have an arbitrary measurable outcome space, and let $Y$ be its
outcome.  For every measurable map $\phi$ with finite range, $\phi(Y)$ is the
outcome of the finite POVM with effects $M(\phi^{-1}(z))$.  The bound just
proved therefore applies to $J_{\phi(Y)}(\rho)$.  Taking the supremum over
$\phi$ in \cref{eq:general-outcome-fi-supremum} proves the same bound for
$J_M(\rho)$.
This establishes \cref{eq:block-fi-eigenvalue}.  Under
\cref{eq:block-conditioning}, its two terms are at most constant multiples
of $d^2t^{3/2}$ and $d^3t$, respectively, which proves
\cref{eq:block-fi-main}.
\end{proof}

It remains to prove the centered average bound.  A direct pointwise estimate
would give $\norm{G_1(u)}_F^2\leq t^2$, which is too large by a factor of
$\sqrt t$.  Instead, the bound on the unfolded matrix from \cite{CLL24} is
linear in the squared projections onto the Schur subspaces.  The relation
$\sum_zp_z\proj{u_z}=\rho^{\otimes t}$ determines the average squared
projection onto each Schur subspace.

\begin{proof}[Proof of \cref{lem:centered-average-g1}]
Averaging the bound on the unfolded matrix
(\cref{fact:cll-unfolded-bound}) and the Schur projection weights for the
pure-state decomposition (\cref{fact:ensemble-schur-projections}) gives
\begin{align}
  \sum_zp_z\norm{G_1(u_z)}_F^2
  &\leq
  \sum_{\substack{\lambda\vdash t\\\ell(\lambda)\leq d}}
  \Paren{\sum_zp_z\norm{\Pi_\lambda u_z}_2^2}
  \sum_i\lambda_i^2
  \notag\\
  &=
  \E_{\lambda\sim\SW^t(\operatorname{spec}\rho)}
  \Brac{\sum_i\lambda_i^2}
  \notag\\
  &\leq
  t(t-1)\Tr(\rho^2)+2t^{3/2}.
\label{eq:uncentered-g1-sharp}
\end{align}
The last line is the sharper form of the Schur--Weyl moment estimate from
\cite[full version, proof of Claim~3.27]{CLL24}, stated in
\cref{fact:sw-second-moment}.

The relation $\sum_zp_zG_1(u_z)=t\rho$ in
\cref{eq:block-ensemble-identities} shows that $G_1(u_z)$ has mean $t\rho$.
Expanding the square around this mean gives
\begin{align*}
  \sum_zp_z\norm{G_1(u_z)-t\rho}_F^2
  &=
  \sum_zp_z\norm{G_1(u_z)}_F^2-t^2\Tr(\rho^2)\\
  &\leq
  2t^{3/2}-t\Tr(\rho^2)\\
  &\leq
  2t^{3/2}.
\end{align*}
This proves \cref{eq:centered-average-g1}.
\end{proof}